\documentclass[12pt,a4paper]{article}

\usepackage{amsmath,amsfonts,amsthm}
\usepackage{mathrsfs}

\makeatletter

 \newcommand{\lyxaddress}[1]{
   \par {\raggedright #1
   \vspace{1.4em}
   \noindent\par}
 }

\makeatother

\newtheorem{thm}{Theorem}
\newtheorem{prop}[thm]{Proposition}
\newtheorem{lem}[thm]{Lemma}
\newtheorem{cor}[thm]{Corollary}

\newtheorem{athm}{Theorem}[section]
\newtheorem{aprop}[athm]{Proposition}
\newtheorem{alem}[athm]{Lemma}

\theoremstyle{remark}
\newtheorem*{rem*}{Remark}
\newtheorem*{rems*}{Remarks}

\theoremstyle{definition}
\newtheorem{define}[thm]{Definition}
\newtheorem*{def*}{Definition}

\newcommand{\sD}{\mathscr{D}}
\newcommand{\sB}{\mathscr{B}}
\newcommand{\sH}{\mathscr{H}}
\newcommand{\sK}{\mathscr{K}}
\newcommand{\sS}{\mathscr{S}}
\newcommand{\bx}{\mathbf{x}}
\newcommand{\bp}{\mathbf{p}}
\renewcommand{\AA}{\mathcal{A}}
\newcommand{\BB}{\mathcal{B}}
\newcommand{\CC}{\mathcal{C}}
\newcommand{\EE}{\mathcal{E}}

\newcommand{\HH}{\mathcal{H}}

\newcommand{\R}{\mathbb{R}}
\newcommand{\Z}{\mathbb{Z}}
\newcommand{\N}{\mathbb{N}}
\newcommand{\C}{\mathbb{C}}
\newcommand{\I}{\mathbb{I}}
\newcommand{\gX}{\mathfrak{X}}
\newcommand{\gS}{\mathfrak{S}}
\newcommand{\gx}{\mathfrak{x}}

\newcommand{\Dom}{\mathop\mathrm{Dom}\nolimits}
\newcommand{\Ker}{\mathop\mathrm{Ker}\nolimits}
\newcommand{\Spec}{\mathop\mathrm{Spec}\nolimits}
\newcommand{\Span}{\mathop\mathrm{span}\nolimits}
\newcommand{\Ran}{\mathop\mathrm{Ran}\nolimits}
\newcommand{\dist}{\mathop\mathrm{dist}\nolimits}
\renewcommand{\Im}{\mathop\mathrm{Im}\nolimits}
\newcommand{\ad}{\mathop\mathrm{ad}\nolimits}
\newcommand{\ii}{\mathrm{i}}
\newcommand{\dd}{\mathrm{d}}

\begin{document}

\title{On the stability of periodically time-dependent quantum systems}


\author{P.~Duclos$^{1}$, E.~Soccorsi$^{1}$,
  P.~\v{S}\v{t}ov\'\i\v{c}ek$^{2}$, M.~Vittot$^{1}$}

\maketitle

\lyxaddress{$^{1}$ Centre de Physique th\'eorique de Marseille UMR
  6207 - Unit\'e Mixte de Recherche du CNRS et des Universit\'es
  Aix-Marseille I, Aix-Marseille II et de l' Universit\'e du Sud
  Toulon-Var - Laboratoire affili\'e \`a la FRUMAM}

\lyxaddress{$^{2}$Department of Mathematics, Faculty of Nuclear
  Science, Czech Technical University, Trojanova 13, 120 00 Prague,
  Czech Republic}

\begin{abstract}\noindent
  The main motivation of this article is to derive sufficient
  conditions for dynamical stability of periodically driven quantum
  systems described by a Hamiltonian $H(t)$, i.e., conditions under
  which it holds
  $\sup_{t\in\R}|\langle\psi_{t},H(t)\psi_{t}\rangle|<\infty$ where
  $\psi_{t}$ denotes a trajectory at time $t$ of the quantum system
  under consideration. We start from an analysis of the domain of the
  quasi-energy operator. Next we show, under certain assumptions, that
  if the spectrum of the monodromy operator $U(T,0)$ is pure point
  then there exists a dense subspace of initial conditions for which
  the mean value of energy is uniformly bounded in the course of time.
  Further we show that if the propagator admits a
  \textit{differentiable} Floquet decomposition then $\|H(t)\psi_t\|$
  is bounded in time for any initial condition $\psi_0$, and one
  employs the quantum KAM algorithm to prove the existence of this
  type of decomposition for a fairly large class of $H(t)$. In
  addition, we derive bounds uniform in time on transition
  probabilities between different energy levels, and we also propose
  an extension of this approach to the case of a higher order of
  differentiability of the Floquet decomposition. The procedure is
  demonstrated on a solvable example of the periodically
  time-dependent harmonic oscillator.
\end{abstract}

\section{Introduction}%

We discuss several topics related to the dynamical properties of
periodically time-dependent quantum systems. Such a system is
described by a Hamiltonian $H(t)$ in a Hilbert space $\sH$ depending
on $t$ periodically with a period $T$, and we suppose that the
propagator $U(t,s)$ associated to the Hamiltonian $H(t)$ exists.

We start our exposition from an analysis of the domain of the Floquet
Hamiltonian (the quasi-energy operator). The quasi-energy operator is
a basic tool in the theory of time-dependent quantum systems and is
closely related to the monodromy operator $U(T,0)$
\cite{howland79:_scatter_hamiltonians_periodic,
  yajima77:_scatter_schroedinger_eqs_periodic}. This is a common
belief that the dynamical properties are essentially determined by the
spectral properties of $U(T,0)$. It is shown in
\cite{enss_veselic83:_bound_propag_states_time_dep_hamil} that $\psi$
belongs to $\sH^{pp}(U(T,0))$ (the subspace in $\sH$ corresponding to
the pure point spectrum of $U(T,0)$) if and only if the trajectory
$\{\psi_t;\textrm{~}t\geq0\}$ is precompact (where
$\psi_t=U(t,0)\psi$). Under the assumptions that $H(0)$ is positive,
discrete and unbounded, and that the perturbation $H(t)-H(0)$ is
uniformly bounded, it is observed in
\cite{deoliveira95:_some_rems_stabil_nonstationary_qs} that if the
mean value of energy, $\langle\psi_t,H(t)\psi_t\rangle$, is bounded
then the corresponding trajectory $\{\psi_t;\textrm{~}t\geq0\}$ is
precompact. Jointly this implies that if the mean value of energy is
bounded for any initial condition then $U(T,0)$ has a pure point
spectrum. To our knowledge, the inverse implication is not clarified
yet. In the present paper we show, under certain assumptions, that if
the spectrum of $U(T,0)$ is pure point then there exists, in $\sH$, a
dense subspace of initial conditions for which the mean value of
energy is bounded. However it has been shown very recently in
\cite{deoliveira_simsen07} that there exist situations when some
trajectories may lead to unbounded energy in spite of pure pointness
of $U(T,0)$.

There is no doubt that the knowledge of evolution of the mean value of
energy in the case of time-dependent systems is important from the
physical point of view. This is also our basic topic in this paper.
More precisely, instead of treating directly the mean value of energy
we consider the quantity $\|H(t)\psi_t\|$. Naturally, this type of
problems attracted attention in the past though the results are less
numerous than one might expect. Let us mention some of them that
motivated us though in no way we attempt to provide an exhaustive
list.

Assuming a growing gap structure of the spectrum of $H(0)$ it is shown
in \cite{nenciu97:_adiabat_th_stabil_increasing_gaps} with the aid of
adiabatic methods that
$\langle\psi_{t},H(t)\psi_{t}\rangle=O(t^{\delta})$ where $\delta>0$
is inversely proportional to the order of differentiability of $H(t)$.
An upper bound of this type is also derived in
\cite{joye96:_upper_bounds_energy_time_depend} under rather mild
assumptions on the gap structure of the spectrum and without
differentiability of $H(t)$. On the other hand, the latter result is
directly applicable only provided the perturbation is in certain sense
small when compared to $H(0)$. For example, in the case of simple
spectrum the operator $H(0)^q(H(t)-H(0))$ is required to be
Hilbert-Schmidt for some $q\geq1/2$. Some extensions and applications
can be also found in
\cite{barbaroux_joye98:_expect_values_timedep_qs}. These estimates on
the growth of energy were derived without assuming the periodicity.
Let us also mention \cite{duclos_lev_ps07} where bounds on the energy
growth are derived in the case of shrinking gaps in the spectrum.

A stronger result is known for periodically time-dependent systems
\cite{asch_duclos_exner98:_stability_growing_gaps_wannier}. It
suggests that for a large class of periodic systems one can expect
uniform boundedness of the mean value of energy for any initial
condition $\psi\in\Dom H(0)$. Further, in \cite{wang07} the energy is
shown to be uniformly bounded in time in the particular case when the
harmonic oscillator is driven by quasi-periodically time dependent
Gaussian potentials for suitable non resonant frequencies and a small
enough coupling constant. It is proposed in
\cite{debievre_forni98:_transport_kicked_quasiperiod_hamil} to call
this property \emph{dynamical stability}. We adopt this terminology in
the current paper.

Though the ideas concerning the dynamical stability are developed in
\cite{asch_duclos_exner98:_stability_growing_gaps_wannier} on a
particular example of the driven ring it is indicated there that they
are valid also under more general settings. The proof is based on two
observations. First, if the propagator admits a differentiable Floquet
decomposition in the sense that it can be written in the form
$U(t)=U_{F}(t)\exp(-\ii tH_{F})$ where $H_{F}$ is self-adjoint and
$U_{F}(t)$ is a periodic and strongly differentiable family of unitary
operators then the system is dynamically stable.  According to the
second observation one can use the quantum KAM
(Kolmogorov-Arnold-Moser) algorithm to show that the propagator
actually admits this type of decomposition in the case when $H(0)$ is
a semi-bounded discrete operator obeying a gap condition, and provided
the frequency is non-resonant and the time-dependent perturbation is
sufficiently small. In particular, the result in
\cite{asch_duclos_exner98:_stability_growing_gaps_wannier} is based on
a formulation of the quantum KAM theorem presented in
\cite{duclos_ps96:_floquet_hamiltonian_pure_point}.

In the current paper we wish to further develop the basic ideas from
\cite{asch_duclos_exner98:_stability_growing_gaps_wannier} and
particularly to work out the proofs in full detail when considering
applications of these ideas to more general systems. In addition we
derive uniform bounds on transition probabilities between different
energy levels. Moreover, we propose an extension to the case when the
Floquet decomposition is $p$ times continuously differentiable in the
strong sense. Restricting the perturbation $V(t)=H(t)-H(0)$ to a
certain class of operator-valued functions by requiring the multiple
commutators with $H(0)$ to be bounded up to some order one can show
that $\|H(t)^p\psi_t\|$ is bounded in time. Furthermore, the basic
procedure is demonstrated on the solvable example of the periodically
time-dependent harmonic oscillator. For the purposes of this example
we collect in the Appendix some useful formulas for the propagator.
Finally we combine the procedure based on the differentiability of the
Floquet decomposition with an improved version of the quantum KAM
theorem that was presented in
\cite{duclos_lev_ps_vittot02:_weakly_regular_floquet_hamil}.

\section{The Floquet Hamiltonian}
\label{sec:floquet_hamiltonian}

Let us make more precise the assumptions on the Hamiltonian. Let
$\{H(t);\textrm{~}t\in\R\}$ be a family of self-adjoint operators such
that the domain $\Dom(H(t))$ does not depend on time. Further we
assume that the propagator $U(t,s)$ associated to $H(t)$ exists. This
means that $U(t,s)$ is a function with values in $\sB(\sH)$ which is
strongly continuous jointly in $t$ and $s$, $U(t,t)=\I$, the domain
$\Dom(H(0))$ is invariant under the action of $U(t,s)$ for all $t$,
$s$, and
\begin{displaymath}
  \forall\psi\in\Dom(H(0)),\textrm{~}\ii\,\partial_{t}U(t,s)\psi
  = H(t)\,U(t,s)\psi.
\end{displaymath}
Then the propagator is unique, unitary and satisfies the
Chapman-Kolmogorov equation: $U(t,r)U(r,s)=U(t,s)$ for all
$t,r,s\in\R$.

Let us recall that usually one imposes a standard sufficient condition
that guarantees the existence of the evolution operator. Namely, if
the mapping
\[
t,s\mapsto\frac{1}{t-s}\,\left((H(t)+\ii)(H(s)+\ii)^{-1}-\I\right)
\]
can be extended for $t=s$ to a strongly continuous mapping
$\R^{2}\to\sB(\sH)$ then the propagator exists
\cite{reed_simon75:methods_modern_mf_II}. For more general sufficient
conditions one can consult the monographs
\cite{simon71:_hamiltonians_as_quadratic_forms} and
\cite{krein71:_linear_diff_eqs_banach_space}. But as already stated,
we assume directly the existence of the propagator without bothering
about particular hypotheses that guarantee it.

Since the Hamiltonian $H(t)$ is assumed to be $T$-periodic the same
is true for the propagator. This means that
\begin{equation}
  \label{eq:U_period}
  \forall t,s,\textrm{~}U(t+T,s+T)=U(t,s).
\end{equation}
Notice also that by the closed graph theorem the operator
$H(t)\left(H(0)+\ii\right)^{-1}$ is bounded. In addition, in this
section we impose the following two assumptions:
\begin{eqnarray}
  & & \R\ni t\mapsto\|H(t)\left(H(0)+\ii\right)^{-1}\|
  \textrm{~is locally bounded,}
  \label{eq:Ht_locbound}\\
  & & \forall\psi\in\Dom(H(0)),\textrm{~}
  \R\ni t\mapsto\|H(0)U(t,0)\psi\|
  \textrm{~is locally square integrable.}
  \label{eq:HUtpsi_locbound}
\end{eqnarray}
In fact, hypothesis (\ref{eq:Ht_locbound}) means that
$H(t)\left(H(0)+\ii\right)^{-1}$ is bounded uniformly in $t$ since we
are considering the periodic case.

An important tool when investigating time dependent quantum systems is
the Floquet Hamiltonian (also called the quasi-energy operator)
\cite{howland79:_scatter_hamiltonians_periodic,
  yajima77:_scatter_schroedinger_eqs_periodic}. It acts in the Hilbert
space
\[
\sK = L^{2}([\,0,T\,],\sH,\dd t)
\equiv L^{2}([\,0,T\,],\dd t)\otimes\sH.
\]
If convenient we shall regard the elements of $\sK$ as $T$-periodic
vector-valued functions on $\R$ with values in $\sH$. A unique Floquet
Hamiltonian is associated to any strongly continuous propagator via
the Stone theorem according to the prescription
\begin{equation}
  \label{eq:def_K}
  \forall f\in\sK,\,\forall\sigma\in\R,
  \textrm{~for~a.a.~}t\in\R\textrm{,~}
  (e^{-\ii\sigma K}f)(t)=U(t,t-\sigma)f(t-\sigma).
\end{equation}
Hence $f$ belongs to $\Dom(K)$ if and only if the derivative
$\ii\partial_{\sigma}U(t,t-\sigma)f(t-\sigma)|_{\sigma=0}$ exists in
$\sK$. Morally the Floquet Hamiltonian can be regarded as
$-\ii\partial_{t}+H(t)$ but in general this formal expression should
be interpreted in a weak sense. The following remarks aim to provide
some details about the definition of $K$.

In the particular case when the Hamiltonian does not depend on time
and equals $H_{0}$ for all $t$ it holds
$U(t,t-\sigma)=\exp(-\ii\sigma{}H_{0})$ and one easily finds from
(\ref{eq:def_K}) that the associated Floquet Hamiltonian $K_{0}$ is
nothing but the closure of the operator
$-\ii\partial_{t}\otimes1+1\otimes H_{0}$ defined on the algebraic
tensor product $\Dom(\ii\partial_{t})\otimes\Dom(H_{0})$. Here and
everywhere in what follows the time derivative is automatically
considered with the periodic boundary conditions. This is to say that
the orthonormal basis
$\{T^{-1/2}\exp(2\pi\ii{}kt/T);\textrm{~}k\in\Z\}$ in
$L^{2}([\,0,T\,],\dd{}t)$ is formed by eigenfunctions of
$\ii\partial_{t}$.

Let us denote by $C_{T}^{\infty}(\R)$ the space of $T$-periodic
smooth functions on $\R$ and let
\[
C_{T}^{\infty}(\R)\otimes\Dom(H(0))
= \Span\{\eta(t)\psi;\textrm{~}
\eta\in C_{T}^{\infty}(\R),\textrm{~}\psi\in\Dom(H(0))\}\subset\sK
\]
be the algebraic tensor product. It is straightforward to see that
$C_{T}^{\infty}(\R)\otimes\Dom(H(0))\subset\Dom(K)$ and
\[
K(\eta\otimes\psi)(t)=-\ii\,\eta'(t)\psi+\eta(t)\, H(t)\psi,
\]
for every $\eta\in C_{T}^{\infty}(\R)$ and $\psi\in\Dom(H(0))$. Set
\[
K^{0}=K\big|_{C_{T}^{\infty}(\R)\otimes\Dom(H(0))}.
\]
It follows that $K^{0}$ is a symmetric operator,
$K^{0}\subset{}K\subset(K^{0})^{\ast}$.

Let $K^{1}$ be another operator acting in $\sK$ and defined by the
prescription: $f\in\Dom(K^{1})$ if and only if, for every
$\psi\in\Dom(H(0))$, the function
$t\mapsto\langle\psi,f(t)\rangle_{\sH}$ is absolutely continuous and
there exists $g\in\sK$ such that
\begin{equation}
  \label{eq:Kf_eq_g_weak}
  \forall\psi\in\Dom(H(0)),\textrm{~}
  -\ii\partial_{t}\langle\psi,f(t)\rangle_{\sH}+\langle
  H(t)\psi,f(t)\rangle_{\sH} = \langle\psi,g(t)\rangle_{\sH},
\end{equation}
(the last equality is valid, of course, almost everywhere on $\R$).
In that case $g$ is unique and we set $K^{1}f=g$.

From the definition it is obvious that $K^{0}\subset K^{1}$. Hence $K$
and $K^{1}$ coincide on $C_{T}^{\infty}(\R)\otimes\Dom(H(0))$. We
shall show that $K$ and $K^{1}$ are actually equal. Let us make a
remark on the notation used below and everywhere in the remainder of
the paper: the natural numbers $\N$ start from $1$ while $\Z_+$ stands
for non-negative integers.

\begin{lem}
  \label{lem:Dt_Utpsi_ft}
  For all $\psi\in\sH$ and $f\in\Dom(K^{1})$, the function
  $\langle{}U(t,0)\psi,f(t)\rangle_{\sH}$ is absolutely continuous and
  it holds true that
  \begin{equation}
    \label{eq:Dt_Utpsi_ft}
    -\ii\partial_{t}\langle{}U(t,0)\psi,f(t)\rangle_{\sH}
    =\langle U(t,0)\psi,g(t)\rangle_{\sH}\textrm{~~for~a.e.~}t\in\R,
  \end{equation}
  where $g=K^{1}f$.
\end{lem}

\begin{proof}
  Let us first suppose that $\psi\in\Dom(H(0))$. Let $P$ be the
  projector-valued measure for $H(0)$ and set $P_{n}=P([-n,n])$,
  $n\in\N$. Then $P_{n}\to\I$ strongly as $n\to\infty$ and therefore
  the following limit is true in the space of distributions $\sD'(\R)$
  (actually in $L_{\mbox{\scriptsize{loc}}}^{1}(\R)$):
  \[
  \lim_{n\to\infty}\,\langle U(t,0)\psi,P_{n}f(t)\rangle_{\sH}
  =\langle U(t,0)\psi,f(t)\rangle_{\sH}.
  \]
  We shall compute the time derivative of
  $\langle{}U(t,0)\psi,f(t)\rangle_{\sH}$ in the sense of
  distributions when making use of the fact that $-\ii\partial_{t}$ is
  continuous on $\sD'(\R)$. Choose an orthonormal basis in $\sH$
  called $\{\varphi_{k}\}$. The series
  \[
  \langle U(t,0)\psi,P_{n}f(t)\rangle_{\sH}
  = \sum_{k}\langle U(t,0)\psi,\varphi_{k}\rangle_{\sH}
  \langle\varphi_{k},P_{n}f(t)\rangle_{\sH}
  \]
  converges in $\sD'(\R)$ (actually in
  $L_{\mbox{\scriptsize{loc}}}^{1}(\R)$) since it converges absolutely
  and is majorized by $\|\psi\|_\sH\|f(t)\|_{\sH}$, a locally
  integrable function. Then, in the sense of distributions,
  \begin{eqnarray}
    & & -\ii\partial_{t}\langle U(t,0)\psi,P_{n}f(t)\rangle_{\sH}
    =\sum_{k}\big(\langle H(t)U(t,0)\psi,\varphi_{k}\rangle_{\sH}
    \langle\varphi_{k},P_{n}f(t)\rangle_{\sH}\nonumber \\
    & & \qquad\qquad\quad -\langle U(t,0)\psi,\varphi_{k}\rangle_{\sH}
    \langle H(t)P_{n}\varphi_{k},f(t)\rangle_{\sH}
    +\langle U(t,0)\psi,\varphi_{k}\rangle_{\sH}
    \langle\varphi_{k},P_{n}g(t)\rangle_{\sH}\big).\nonumber \\
    \label{eq:Dt_Utpsift_sum_phik}
  \end{eqnarray}
  Here we have used the definition of $K^{1}$ (note that
  $P_n\varphi_k\in\Dom(H(0))$).

  The RHS in (\ref{eq:Dt_Utpsift_sum_phik}) splits into three sums
  each of them can be summed in $\sD'(\R)$. To see it let us note that
  with the aid of the Schwarz inequality and the Parseval equality one
  can estimate
  \begin{equation}
    \label{eq:sumk_HUpsi_phik_Pf}
    \sum_{k}\left|\langle H(t)U(t,0)\psi,\varphi_{k}\rangle_{\sH}
      \langle\varphi_{k},P_{n}f(t)\rangle_{\sH}\right|
    \leq\| H(t)U(t,0)\psi\|_{\sH}\|f(t)\|_{\sH}.
  \end{equation}
  Furthermore, $\|f(t)\|_{\sH}$ is square integrable and
  \[
  \|H(t)U(t,0)\psi\|_{\sH} \leq \|H(t)\left(H(0)+\ii\right)^{-1}\|
  \|\left(H(0)+\ii\right)U(t,0)\psi\|_{\sH}
  \]
  is locally square integrable due to (\ref{eq:Ht_locbound}) and
  (\ref{eq:HUtpsi_locbound}). Hence the RHS of
  (\ref{eq:sumk_HUpsi_phik_Pf}) is locally integrable. As far as the
  second sum is concerned let us note that
  \[
  G_{n}(t):=H(t)P_{n}=H(t)\left(H(0)+\ii\right)^{-1}
  \left(H(0)+\ii\right)P_{n}
  \]
  is a bounded operator and even $\|G_{n}(t)\|$ is locally bounded
  according to hypothesis (\ref{eq:Ht_locbound}). Finally, the third
  sum does not cause any problem. Consequently, the RHS of
  (\ref{eq:Dt_Utpsift_sum_phik}) equals
  \begin{equation}
    \label{eq:Dt_Utpsift_summed}
    \langle H(t)U(t,0)\psi,P_{n}f(t)\rangle_{\sH}
    -\langle U(t,0)\psi,G_{n}(t)^{\ast}f(t)\rangle_{\sH}
    +\langle U(t,0)\psi,P_{n}g(t)\rangle_{\sH}\,.
  \end{equation}
  Thus $-\ii\partial_{t}\langle U(t,0)\psi,f(t)\rangle_{\sH}$ is equal
  to the limit of (\ref{eq:Dt_Utpsift_summed}) as $n\to\infty$.

  Since for every $\varphi\in\Dom(H(0))$ it holds
  $H(0)P_{n}\varphi\to{}H(0)\varphi$ and $U(t,0)\psi\in\Dom(H(0))$, in
  the second term in (\ref{eq:Dt_Utpsift_summed}) we get
  \[
  \lim_{n\to\infty}G_{n}(t)U(t,0)\psi
  = \lim_{n\to\infty}H(t)\left(H(0)+\ii\right)^{-1}\left(H(0)
    +\ii\right)P_{n}U(t,0)\psi=H(t)U(t,0)\psi.
  \]
  The point-wise limits of the first and the third term in
  (\ref{eq:Dt_Utpsift_summed}) are obvious. To justify the convergence
  in $\sD'(\R)$ one can apply once more assumptions
  (\ref{eq:Ht_locbound}) and (\ref{eq:HUtpsi_locbound}) to show that
  each term has a locally integrable majorant which is independent of
  $n$. Thus sending $n\to\infty$ one finds that equality
  (\ref{eq:Dt_Utpsi_ft}) holds true in the sense of distributions.
  The RHS is a locally integrable function. By a standard result of
  the theory of distributions this implies that the function
  $\langle{}U(t,0)\psi,f(t)\rangle_{\sH}$ is absolutely continuous and
  that equality (\ref{eq:Dt_Utpsi_ft}) holds true in the usual sense.

  Finally let us show that the condition $\psi\in\Dom(H(0))$ from the
  beginning of the proof can be relaxed. Actually, if $h\in\sK$ and
  $\psi_k\to\psi$ in $\sH$ then
  $\langle{}U(t,0)\psi_k,h(t)\rangle_\sH$ is locally integrable and
  this sequence of functions converges to
  $\langle{}U(t,0)\psi,h(t)\rangle_\sH$ in the $L^1$ norm on every
  bounded interval and hence in the sense of distributions. For any
  $\psi\in\sH$ one can choose a sequence $\psi_k\in\Dom(H(0))$ such
  that $\psi_k\to\psi$ and then send $k\to\infty$ in the equality
  \begin{displaymath}
    -\ii\partial_{t}\langle{}U(t,0)\psi_k,f(t)\rangle_{\sH}
    =\langle U(t,0)\psi_k,g(t)\rangle_{\sH}\textrm{~~in~}\sD'(\R).
  \end{displaymath}
  Since the function $\langle{}U(t,0)\psi,g(t)\rangle_{\sH}$ is
  locally integrable the function
  $\langle{}U(t,0)\psi,f(t)\rangle_{\sH}$ can be redefined on a
  measure zero set so that it is absolutely continuous and equality
  (\ref{eq:Dt_Utpsi_ft}) holds true in the usual sense.
\end{proof}

\begin{lem}
  \label{lem:K1-symmetric}
  $K^{1}$ is symmetric.
\end{lem}

\begin{proof}
  Suppose that $f\in\Dom(K^{1})$, $K^{1}f=g$ and $\psi\in\sH$.
  According to Lemma~\ref{lem:Dt_Utpsi_ft} we have
  \[
  -\ii\partial_{t}\left|\langle U(t,0)\psi,f(t)\rangle_{\sH}\right|^{2}
  =2\ii\Im\!\left(\langle f(t),U(t,0)\psi\rangle_{\sH}
    \langle U(t,0)\psi,g(t)\rangle_{\sH}\right).
  \]
  Let $\{\psi_{k}\}$ be an orthonormal basis in $\sH$. Then, for
  almost all $s\in\R$ and all $k$,
  \begin{eqnarray*}
    &  & \left|\langle\psi_{k},U(s+T,0)^{-1}f(s+T)\rangle_{\sH}\right|^{2}
    -\left|\langle\psi_{k},U(s,0)^{-1}f(s)\rangle_{\sH}\right|^{2}\\
    &  & \qquad\qquad=-2\Im\int_{s}^{s+T}
    \langle U(t,0)^{-1}f(t),\psi_{k}\rangle_{\sH}
    \langle\psi_{k},U(t,0)^{-1}g(t)\rangle_{\sH}\,\dd t\,.
  \end{eqnarray*}
  Summing in $k$ one can commute the sum and the integral.
  Consequently, for almost all $s$,
  \[
  \|f(s+T)\|_{\sH}^{\,2}-\|f(s)\|_{\sH}^{\,2}
  = -2\Im\int_{s}^{s+T}\langle f(t),g(t)\rangle_{\sH}\,\dd t
  = -2\Im(\langle f,g\rangle_{\sK}).
  \]
  Since $\|f(t)\|_{\sH}$ is periodic the LHS vanishes almost
  everywhere. We find that $\forall f\in\Dom(K^{1})$,
  $\Im(\langle{}f,K^{1}f\rangle_{\sK})=0$. This shows that $K^{1}$ is
  symmetric.
\end{proof}

\begin{lem}\label{prop:K0ast_eq_K1}
  $(K^{0})^{\ast}=K^{1}$. Consequently, $K^{1}$ is closed and $K^{0}$
  is essentially self-adjoint.
\end{lem}

\begin{proof}
  By definition, $f\in\Dom((K^{0})^{\ast})$ if and only if there
  exists $g\in\sK$ such that
  \[
  \forall h\in\Dom(K^{0}),\textrm{~}\langle K^{0}h,f\rangle_{\sK}
  =\langle h,g\rangle_{\sK}.
  \]
  Moreover, in that case $g$ is unique and $(K^{0})^{*}f=g$. Setting
  $h=\eta\otimes\psi$ we find that for all $\psi\in\Dom(H(0))$ it is
  true that
  \[
  \forall\eta\in C_{T}^{\infty}(\R),\textrm{~}
  \int_{0}^{T}\left(\ii\eta'(t)\langle\psi,f(t)\rangle_{\sH}
    +\eta(t)\langle H(t)\psi,f(t)\rangle_{\sH}\right)\dd t
  =\int_{0}^{T}\eta(t)\langle\psi,g(t)\rangle_{\sH}\dd t\,.
  \]
  The last statement can be rewritten as equality
  (\ref{eq:Kf_eq_g_weak}) valid in the sense of distributions. Since
  the both functions $\langle H(t)\psi,f(t)\rangle_{\sH}$ and
  $\langle\psi,g(t)\rangle_{\sH}$ belong to
  $L_{\mbox{\scriptsize{loc}}}^{1}(\R)$ (using again
  (\ref{eq:Ht_locbound}) in the former case) the standard results of
  the theory of distributions tell us that
  $\langle\psi,f(t)\rangle_{\sH}$ is actually absolutely continuous
  and equality (\ref{eq:Kf_eq_g_weak}) holds true in the usual sense.
  Thus we conclude that $f\in\Dom(K^{1})$ and $K^{1}f=g$. Hence
  $(K^{0})^{\ast}\subset{}K^{1}$. Now it suffices to apply
  Lemma~\ref{lem:K1-symmetric}. Actually, the relations
  \[
  (K^{0})^{\ast}\subset K^{1}\subset(K^{1})^{\ast}
  \subset (K^{0})^{\ast\ast}
  =\overline{K^{0}}\subset(K^{0})^{\ast}
  \]
  imply that $K^1=(K^{0})^{\ast}$ is closed and
  $(\overline{K^{0}})^{\ast}=(K^{0})^{\ast}=\overline{K^{0}}$.
\end{proof}

\begin{prop}
  \label{prop:K_eq_K1}
  Assuming (\ref{eq:Ht_locbound}) and (\ref{eq:HUtpsi_locbound}), it
  holds true that
  \[
  K=K^{1}=\overline{K^{0}}.
  \]
  In particular, $C_{T}^{\infty}(\R)\otimes\Dom(H(0))$ is a core of
  $K$.
\end{prop}

\begin{proof}
  According to Lemma~\ref{lem:K1-symmetric} and
  Lemma~\ref{prop:K0ast_eq_K1} it holds true that
  \[
  K^{0}\subset\overline{K^{0}}\subset K
  =K^{\ast}\subset(K^{0})^{\ast}=\overline{K^{0}}
  \]
  and $K^{1}=(K^{0})^{\ast}=\overline{K^{0}}$. The proposition follows
  immediately.
\end{proof}

Let us note that if a vector-valued function $f(t)$ from the domain of
$K$ is even known to be continuously differentiable (in the strong
sense) then necessarily $f(t)\in\Dom(H(0))$ for all $t$. Under this
additional assumption we actually have
\[
(Kf)(t)=-\ii\partial_{t}f(t)+H(t)f(t)=g(t).
\]
In the general case, however, one should use the weaker form
(\ref{eq:Kf_eq_g_weak}). The relation between $K$ and the formal
expression $-\ii\partial_{t}+H(t)$ can be also expressed as follows.
Let $H=\int^{\oplus}H(t)\,\dd{}t$ be the self-adjoint operator in
$\sK$ with the domain formed by those $f\in\sK$ satisfying
$f(t)\in\Dom(H(0))$ for a.a. $t$ and $\int_{0}^{T}\Vert
H(t)f(t)\Vert^{2}\dd t<\infty$, with $(Hf)(t)=H(t)f(t)$. Clearly,
\[
\Dom(K)\supset\Dom(-\ii\partial_{t}\otimes1)\cap\Dom(H)
\supset C_{T}^{\infty}(\R)\otimes\Dom(H(0))
\]
and therefore, according to Proposition~\ref{prop:K_eq_K1},
$\Dom(-\ii\partial_{t}\otimes1)\cap\Dom(H)$ is a core of $K$.
Hence
\begin{equation}
  \label{eq:K_eq_dt_plus_H}
  K=\overline{-\ii\partial_{t}\otimes1+H}\,.
\end{equation}

\section{Boundedness of energy for a dense set of initial conditions}

In this section we consider slightly more general periodically
time-dependent Hamiltonians $H(t)$, $t \in \R$, than those presented
in the beginning of Section~\ref{sec:floquet_hamiltonian}, at least
among those which are bounded below. We suppose that the Hamiltonian
$H(t)$ is associated to a closed, densely defined and positive
sesquilinear form $q(t)$, with a domain independent of $t$:
\begin{equation}
  \label{DSIC_H1}
  \Dom q(t) = \Dom q(0),\ \forall t \in \R.
\end{equation}
Assuming that the spectrum of $U(T,0)$ is pure point we wish to
construct a rich set of initial conditions for which the mean value of
energy is uniformly bounded in time. It turns out that this is
possible if the eigenvectors of $U(T,0)$ belong to the form domain
$\Dom{}q(0)$.

The space $\Dom q(0)$ endowed with the scalar product
$\langle{}u,v\rangle_1=\langle{}u,v\rangle_\sH+q(0)(u,v)$ is a Hilbert
space denoted by $\sH_1$, and we recall that
\begin{displaymath}
  \langle H(t) u,v\rangle_{\sH}
  = q(t)(u,v),\ \forall u\in\Dom H(t),\ \forall v\in\sH_1,
\end{displaymath}
where
\begin{displaymath}
  \Dom H(t) = \{u\in\sH_1;\textrm{~}
  \exists C_u\geq0 \textrm{~s.t.~} \forall v\in\sH_1,\textrm{~}
  |q(t)(u,v)| \leq C_u\|v\|_{\sH}\}.
\end{displaymath}
We call $\sH_{-1}$ the dual space of $\sH_1$, that is to say the
vector space of continuous conjugate linear forms on $\sH_1$. For any
$u\in\sH$, the functional $v\mapsto\langle{}v,u\rangle_{\sH}$ belongs
to $\sH_{-1}$ since
$|\langle{}v,u\rangle_{\sH}|\leq\|u\|_{\sH}\|v\|_{\sH}
\leq\|u\|_{\sH}\|v\|_1$,
and we can also regard $\sH$ as a subspace of $\sH_{-1}$ with
\[
\|u\|_{-1} = \sup_{v\in\sH_1,\,v\neq0}
\frac{|\langle u,v\rangle_{\sH}|}{\|v\|_1} \leq \|u\|_{\sH}.
 \]
Thus
\begin{displaymath}
  \sH_1 \subset \sH \subset \sH_{-1},
\end{displaymath}
where the symbol $\subset$ means a topological embedding. Actually,
$H(t)$ can be extended into an operator mapping $\sH_1$ into
$\sH_{-1}$ provided there exists a constant $C_t\geq0$ such that
\begin{displaymath}
  \forall u\in\sH_1,\ q(t)(u,u) \leq C_t \|u\|_{1}^{\,2}.
\end{displaymath}
Let us denote by $\langle\cdot,\cdot\rangle_{-1,1}$ the dual pairing
between $\sH_{-1}$ and $\sH_{1}$. This pairing is conjugate linear in
the first and linear in the second argument. In other words, the
embedding $\sH\subset\sH_{-1}$ means that
$\langle\psi,g\rangle_{-1,1}=\langle\psi,g\rangle_\sH$ for all
$\psi\in\sH$ and $g\in\sH_1$, and the mapping
$H(t):\sH_{1}\to\sH_{-1}$ is defined so that
$\langle{}H(t)u,v\rangle_{-1,1}=q(t)(u,v)$ for all $u,v\in\sH_1$.

In the remainder of this section, we will refer to the propagator
$U(t,0)$ associated to the family of Hamiltonians $H(t)$, $t\in\R$.
Its existence is implied by the following result which can be found in
\cite[Theorem~II.27]{simon71:_hamiltonians_as_quadratic_forms} and
that we reproduce below for the reader's convenience.

\begin{thm}
  \label{thm_EPS_weak}
  We assume that $q(t)$ satisfies (\ref{DSIC_H1}) and that there is a
  constant $C\geq1$ such that the operator $H(t)$ satisfies, for all
  $t\in\R$:
  \begin{enumerate} 
  \item $C^{-1} (H(0)+1) \leq H(t) \leq C (H(0)+1)$.
  \item The derivative $\frac{\dd}{\dd{}t} H(t)^{-1}$ exists in the
    norm sense and
    \begin{displaymath}
      \left\|\sqrt{H(t)}\left(\frac{d}{dt}H(t)^{-1}
        \right)\!\sqrt{H(t)}\right\| \leq C.
    \end{displaymath}
  \end{enumerate}
  Then, for any $\psi_0\in\sH_1$ there is a unique function
  $\R\ni{}t\mapsto\psi(t)\in\sH_1$ such that:
  \begin{enumerate}
  \item $\psi$ is $\sH_1$-weakly continuous, i.e., for all
    $g\in\sH_{-1}$, $t\mapsto\langle{}g,\psi(t)\rangle_{-1,1}$ is a
    continuous function.
  \item $\psi$ is a weak solution of the Schr\"odinger equation in the
    following sense:
    \[
    \forall g\in\sH_1,\textrm{~}
    -\ii\,\frac{\dd}{\dd t}\langle g,\psi(t)\rangle_{\sH}
    + q(t)(g,\psi(t)) = 0\textrm{~and~}\psi(0) = \psi_0.
    \]
  \item For all $s \in \R$ we have
    \begin{displaymath}
      \lim_{t\rightarrow s} \left\|\frac{\psi(t)-\psi(s)}{t-s}
        +\ii H(t)\psi(t) \right\|_{-1} = 0.
    \end{displaymath}
  \item $\|\psi(t)\|_{\sH}=\|\psi_0\|_{\sH}$ for all $t\in\R$ and
    $t\mapsto\psi(t)$ is continuous in the norm topology in $\sH$.
  \end{enumerate}
\end{thm}

The propagator $U:(s,t)\in\R^2\mapsto{}U(t,s)$ associated to the
Hamiltonian $H(t)$ is defined by $U(t,s)\psi(s)=\psi(t)$. It is
unitary and strongly continuous according to point~4.

For the proof of the main result of this section,
Proposition~\ref{prop:dense_set_weak}, we need the following lemma.

\begin{lem}
  \label{lm_F_bounded}
  Let $\psi\in\sH_1$ be an eigenfunction of the Floquet operator
  $U(T,0)$. Then the function
  \begin{displaymath}
    F_{\psi}(t) := \|H(t)U(t,0)\psi\|_{-1}
  \end{displaymath}
  is bounded in $\R$:
  \begin{displaymath}
    \|F_{\psi}\|_{\infty} := \sup_{t \in \R} F_{\psi}(t) < +\infty.
  \end{displaymath}
\end{lem}

\begin{proof}
  First, we notice that the function
  $t\mapsto|\langle{}H(t)U(t,0)\psi,g\rangle_{-1,1}|$, with
  $g\in\sH_1$, is periodic with the period $T$. This can be seen from
  the equality
  \[
  U(t+T,0)\psi = U(t+T,T)U(T,0)\psi = \lambda U(t,0)\psi,
  \]
  and
  \begin{displaymath}
    |\langle{}H(t+T)U(t+T,0)\psi,g\rangle_{-1,1}|
    = |\lambda||\langle H(t+T)U(t,0)\psi,g\rangle_{-1,1}|
    = |\langle{}H(t)U(t,0)\psi,g\rangle_{-1,1}|
  \end{displaymath}
  since $|\lambda |=1$ and $H(t+T)=H(t)$.

  Moreover, the $\sH_{-1}$-valued function $t\mapsto{}H(t)U(t,0)\psi$
  is weakly continuous on $\R$. Indeed, for any given real numbers $s$
  and $t$, we derive from the following obvious decomposition, with
  $g\in\sH_1$,
  \begin{eqnarray*}
    && \langle H(t)\psi(t),g\rangle_{-1,1}
    - \langle H(s)\psi(s),g\rangle_{-1,1} \\
    && =\, \Big\langle H(t)\psi(t)
      - \ii\,\frac{\psi(t)-\psi(s)}{t-s},g\Big\rangle_{-1,1}
    - \Big\langle H(s)\psi(s)
      - \ii\,\frac{ \psi(t)-\psi(s)}{t-s},g\Big\rangle_{-1,1},
  \end{eqnarray*}
  that
  \begin{eqnarray*}
    & & |\langle H(t)\psi(t),g\rangle_{-1,1}
    - \langle H(s) \psi(s) , g \rangle_{-1,1}  | \\
    && \leq\, \Big\|H(t)\psi(t)
    - \ii\,\frac{\psi(t)-\psi(s)}{t-s}\Big\|_{-1}\|g\|_1
    + \Big|q(s)(\psi(s),g)
    - \ii\,\frac{\langle\psi(t),g\rangle_\sH
      - \langle\psi(s),g\rangle_\sH}{t-s}\Big|.
  \end{eqnarray*}
  Applying respectively points~3 and 2 of Theorem~\ref{thm_EPS_weak}
  one finds that the both terms on the RHS of the preceding inequality
  tend to zero as $t$ tends to $s$.

  This implies that for every $g\in\sH_1$ the function
  $t\mapsto|\langle{}H(t)U(t,0)\psi,g\rangle_{-1,1}|$ is bounded on
  $\R$ (since we just check that it is periodic). From the uniform
  boundedness principle it follows that
  \begin{displaymath}
    F_\psi(t) = \sup_{g\in\sH_1,\ \|g\|_1=1}|
    \langle{}H(t)U(t,0)\psi,g\rangle_{-1,1}|
  \end{displaymath}
  is bounded on $\R$ as well.
\end{proof}

\begin{prop}
  \label{prop:dense_set_weak}
  Let us suppose that the Floquet operator $U(T,0)$ has a pure point
  spectrum and admits a basis $\mathcal{B}$ formed by eigenfunctions
  belonging to $\sH_1$. Then the energy of the quantum system, when
  starting from any initial state $\psi\in\Span\mathcal{B}$, the set
  of finite linear combinations of vectors from $\mathcal{B}$, is
  bounded in the course of time:
  \[
  \sup_{t\in\R} |\langle H(t)U(t,0)\psi,U(t,0)\psi\rangle_{-1,1}|
  = \sup_{t\in\R} |q(t)\big(U(t,0)\psi,U(t,0)\psi\big)|
  < \infty.
  \]
\end{prop}

\begin{proof}
  Recall that by our assumptions $H(t)^{-1}$ is a bounded operator on
  $\sH$ (see Theorem~\ref{thm_EPS_weak}). If $u\in\sH_1$, $v\in\sH$,
  then $H(t)^{-1}v\in\Dom{}H(t)$ and
  $q(t)(u,H(t)^{-1}v)=\langle{}u,v\rangle_\sH$. Consequently,
  \begin{displaymath}
    \langle H(t)u,H(t)^{-1}v\rangle_{-1,1}
    = q(t)\big(u,H(t)^{-1}v\big) = \langle u,v\rangle_\sH.
  \end{displaymath}

  We can assume that the basis $\BB$ is orthonormal. For any given
  $\psi$ in $\BB\subset\sH_1$, we first notice that $\|U(t,0)\psi\|_1$
  is bounded by $F_{\psi}(t)$ defined in Lemma~\ref{lm_F_bounded} up
  to a multiplicative constant $C$. Indeed, for any $g\in\sH$ we have
  \begin{equation}
    \label{normeres}
    |\langle U(t,0)\psi,g\rangle_{\sH}|
    = |\langle H(t)U(t,0)\psi,H(t)^{-1}g\rangle_{-1,1}|
    \leq \|H(t)U(t,0)\psi\|_{-1}\|H(t)^{-1}g\|_1,
  \end{equation}
  with
  \begin{displaymath}
    \| H(t)^{-1}g \|_1^2
    = \langle H(t)^{-1}g,(H(0)+1)H(t)^{-1} g \rangle_{\sH}
    \leq C \langle H(t)^{-1} g, g \rangle_{\sH}
    \leq C \|H(t)^{-1} g\|_1 \|g\|_{-1},
  \end{displaymath}
  according to assumption~1 in Theorem~\ref{thm_EPS_weak}. Thus
  $\|H(t)^{-1}g\|_1\leq C\|g\|_{-1}$ and (\ref{normeres}) becomes
  \begin{equation}
    \label{normeres2}
    |\langle g,U(t,0)\psi\rangle_{-1,1}|
    = |\langle U(t,0)\psi,g\rangle_{\sH}|
    \leq C\|H(t)U(t,0)\psi\|_{-1}\|g\|_{-1}
    = C F_{\psi}(t)\|g\|_{-1}.
  \end{equation}
  Furthermore, $\sH$ being dense in $\sH_{-1}$ in the norm topology,
  inequality (\ref{normeres2}) remains valid for any $g\in\sH_{-1}$
  implying
  \begin{equation}
    \label{normepsit}
    \|U(t,0)\psi\|_1 = \sup_{g\in\sH_{-1},\ g\neq 0}
    \frac{|\langle g,U(t,0)\psi\rangle_{-1,1}|}{\|g\|_{-1}}
    \leq C F_{\psi}(t).
  \end{equation}

  To complete the proof we pick a function $\varphi$ in
  $\Span\mathcal{B}$, $\varphi=\sum_{i=1}^N c_i\psi_i$, with
  $\psi_i\in\mathcal{B}$ and $c_i\in\C$ for $i=1,2,\ldots,N$. The
  energy function of the quantum system with the initial condition
  $\varphi$ decomposes as
  \[
  E_{\varphi}(t) \equiv
  \langle H(t)U(t,0)\varphi, U(t,0)\varphi\rangle_{-1,1}
  = \sum_{i,j=1}^N c_i \overline{c_j}\,
  \langle H(t)U(t,0)\psi_i, U(t,0)\psi_j\rangle_{-1,1}.
  \]
  Therefore,
  \begin{eqnarray*}
    |E_{\varphi}(t)| & \leq & \sum_{i,j=1}^N |c_i| |c_j|
    \| H(t)U(t,0)\psi_i \|_{-1} \|U(t,0)\psi_j\|_1 \\
    & \leq & C \sum_{i,j=1}^N |c_i||c_j| F_{\psi_i}(t)F_{\psi_j}(t)
    \,\leq\, C\max_{1\leq i\leq N}F_{\psi_i}(t)^{2}
    \left(\sum_{i=1}^N |c_i|\right)^{\!2},
  \end{eqnarray*}
  according to (\ref{normepsit}), so we finally obtain
  \[
  |E_{\varphi}(t)| \leq CN \|\varphi\|^2 \max_{1\leq i\leq N}
  \|F_{\psi_i}\|_{\infty}^2,
  \]
  by combining the Cauchy-Schwarz inequality with Lemma
  \ref{lm_F_bounded}.
\end{proof}

\section{Bounds on energy  and transition probabilities}%

The only assumptions needed in this section are that the domain
$\Dom{}H(t)$ of a $T$-periodic family of self-adjoint operators is
time-independent and that the propagator $U(t,s)$ associated to $H(t)$
exists in the usual sense, as recalled in the beginning of
Section~\ref{sec:floquet_hamiltonian}.

By the spectral theorem, the Floquet (monodromy) operator $U(T,0)$ can
be written in the form $U(T,0)=\exp(-\ii\, TH_{F})$ where $H_{F}$ is a
self-adjoint operator.  Of course, the choice of $H_{F}$ is highly
ambiguous. Let $U_{F}(t)$ be the family of unitary operators defined
by the equality
\begin{equation}
  \label{eq:Floquet-decomp}
  U(t,0)=U_{F}(t)e^{-\ii tH_{F}}.
\end{equation}
Then $U_{F}(0)=\I$ and from the periodicity of $U(t,s)$ (see
(\ref{eq:U_period})) it follows that $U_{F}(t)$ also depends on $t$
periodically. Relation (\ref{eq:Floquet-decomp}) is known as the
Floquet decomposition.

\begin{define}
  \label{def:Floquet_decomp}
  We shall say that a Floquet decomposition is $r$ times continuously
  differentiable in the strong sense for some $r\in\N$ if this is case
  for the family $U_{F}(t)$.

  Furthermore, we shall say that a Floquet decomposition is relatively
  continuously differentiable in the strong sense if the family
  $U_{F}(t)(H_F+\ii)^{-1}$ is continuously differentiable in the
  strong sense. Equivalently this means that for all
  $\psi\in\Dom{}H_F$ the vector-valued function $U_F(t)\psi$ is
  continuously differentiable.
\end{define}

Assume that the propagator $U(t,s)$ admits a Floquet decomposition
which is relatively continuously differentiable in the strong sense.
Set
\begin{equation}
  \label{eq:S_eq_diff_UF}
  S_F(t) = \ii\, U_{F}(t)^{-1}\partial_{t}U_{F}(t),\textrm{~~}
  \Dom{}S_F(t) = \Dom{}H_F.
\end{equation}
By the uniform boundedness principle, $S_F(t)$ is $H_F$--bounded for
all $t\in\R$.  Using the periodicity of $U_{F}(t)$ and applying again
the uniform boundedness principle one finds that
$S_F(t)(H_F+\ii)^{-1}$ is bounded uniformly in $t$. Moreover, $S_F(t)$
is a symmetric operator.

If the Floquet decomposition is even continuously differentiable in
the strong sense then $S_F(t)$ will be naturally supposed to be
defined on the entire space $\sH$. Referring again to the uniform
boundedness principle, in this case we have $S_F(t)\in\sB(\sH)$. Using
the periodicity of $U_{F}(t)$ and applying the uniform boundedness
principle once more one finds that $S_F(t)$ is bounded uniformly in
$t$. Hence $S_F:=\int^{\oplus}S_F(t)\,\dd{}t$ is a bounded operator in
$\sK$ whose norm equals
\[
\|S_F\|=\sup_{t\in\R}\|S_F(t)\|.
\]
Moreover, $S_F(t)$ is a Hermitian operator.

\begin{lem}
  \label{lem:H_eq_S_plus_HAsendw}
  Assume that a Floquet decomposition (\ref{eq:Floquet-decomp}) is
  relatively continuously differentiable in the strong sense and that
  the relative bound of $S_F(t)$ with respect to $H_F$ is less than
  one for all $t$. Then
  \begin{equation}
    \label{eq:H_eq_S_plus_HAsendw}
    \forall t\in\R,\textrm{~}H(t)
    = U_{F}(t)\big(H_{F}+S_F(t)\big)U_{F}(t)^{-1}.
  \end{equation}
  In particular,
  \begin{equation}
    \label{eq:H0_eq_S0_plus_HF}
    H(0) = H_{F}+S_F(0).
  \end{equation}
  Furthermore,
  \[
  \Dom(H_{F})=\Dom(H(0))
  \]
  and this domain is $U_F(t)$ invariant.
\end{lem}

\begin{proof}
  By the assumptions and the Kato-Rellich theorem (see, for example,
  \cite{reed_simon75:methods_modern_mf_II}),
  \begin{displaymath}
    \widetilde{H}(t) := U_{F}(t)\big(H_{F}+S_F(t)\big)U_{F}(t)^{-1},
    \textrm{~~}\Dom{}\widetilde{H}(t)=U_F(t)(\Dom{}H_F),
  \end{displaymath}
  is a $T$-periodic family of self-adjoint operators. From
  (\ref{eq:Floquet-decomp}) it follows that
  \begin{displaymath}
    U(t,s) = U_F(t)e^{-\ii(t-s)H_F}U_F(s)^{-1}.
  \end{displaymath}
  From this relation it is obvious that
  \begin{displaymath}
    U(t,s)(\Dom{}\widetilde{H}(s)) = \Dom{}\widetilde{H}(t).
  \end{displaymath}
  Suppose that $\varphi\in{}U_F(s)(\Dom{}H_F)$ and thus
  $\varphi=U_F(s)\psi$ for some $\psi\in\Dom{}H_F$. A straightforward
  computation yields
  \begin{displaymath}
    \ii\,\partial_tU(t,s)\varphi
    = \ii\,\partial_t\big(U_F(t)e^{-\ii(t-s)H_F}\psi\big)
    = U_F(t)(H_F+S_F(t))e^{-\ii(t-s)H_F}\psi
    = \widetilde{H}(t)U(t,s)\varphi.
  \end{displaymath}
  Hence $U(t,s)$ is a propagator associated to the family
  $\widetilde{H}(t)$.

  Using the property of self-adjointness one can easily see that the
  uniqueness of the relation between a Hamiltonian and a propagator
  applies also in the following direction: if two (in general
  time-dependent) Hamiltonians generate the same propagator then they
  are equal. In our case this means that $\widetilde{H}(t)=H(t)$ for
  all $t$, i.e., equality (\ref{eq:H_eq_S_plus_HAsendw}) holds true.
  Consequently, $U_F(t)(\Dom{}H_F)=\Dom{}H(t)=\Dom{}H(0)$ and setting
  $t=0$ we have $\Dom{}H_F=\Dom{}H(0)$.
\end{proof}

Next we shall show that the relative continuous differentiability of
$U_{F}(t)$ implies the dynamical stability.

\begin{prop}
  \label{prop:energy_bounded}
  Under the same assumptions as in
  Lemma~\ref{lem:H_eq_S_plus_HAsendw}, the energy of the system
  described by the Hamiltonian $H(t)$ is uniformly bounded for any
  initial condition. More precisely,
  \[
  \forall\psi\in\Dom(H(0)),\textrm{~}
  \sup_{t\in\R}\Vert H(t)U(t,0)\psi\Vert \leq C_\psi
  \]
  where
  \begin{displaymath}
    C_\psi = \Vert H_F\psi\Vert
    +\sup_t\Vert S_F(t)(H_F+\ii)^{-1}\Vert\,
    \Vert(H_F+\ii)\psi\Vert.
  \end{displaymath}
\end{prop}

\begin{proof}
  From equalities (\ref{eq:Floquet-decomp}) and
  (\ref{eq:H_eq_S_plus_HAsendw}) it follows that
  \[
  \hspace{7em}
  \Vert H(t)U(t,0)\psi\Vert
  = \Vert(H_F+S_F(t))e^{-\ii tH_F}\psi\Vert
  \leq C_\psi.
  \hspace{7em}\qed
  \]
  \renewcommand{\qed}{}
\end{proof}

\begin{rem*}
  Proposition~\ref{prop:energy_bounded} even implies that the mean
  value of the square of energy, $H(t)^2$, is uniformly bounded.
\end{rem*}

Another application is an estimate of transition probabilities under
the assumption of the strong differentiability of $U_F(t)$. To this
end we shall need the following lemma.

\begin{lem}
  \label{lem:AX-XB_eq_Y}
  Assume that $X,Y\in\sB(\sH)$, $A$ and $B$ are bounded Hermitian
  operators on $\sH$ such that
  \begin{equation}
    \label{eq:AX-XB_eq_Y}
    AX-XB = Y.
  \end{equation}
  If there exist two disjoint closed intervals containing respectively
  $\Spec(A)$ and $\Spec(B)$ then
  \[
  \Vert X\Vert \leq \frac{\Vert{}Y\Vert}{\dist(\Spec(A),\Spec(B))}\,.
  \]
\end{lem}

\begin{proof}
  For the sake of definiteness let us suppose that
  $\inf\Spec(B)>\sup\Spec(A)$.  The solution $X$ of equation
  (\ref{eq:AX-XB_eq_Y}) is unique and given by the
  formula
  \begin{equation}
    \label{eq:inverse_comm_as_int}
    X=\frac{1}{2\pi\ii}\oint_{\gamma}(A-z)^{-1}Y(B-z)^{-1}dz\,.
  \end{equation}
  After a usual limit procedure we can choose for the integration path
  $\gamma$ in (\ref{eq:inverse_comm_as_int}) the line which is
  parallel to the imaginary axis and intersects the real axis in the
  point $(\sup\Spec(A)+\inf\Spec(B))/2$. Integral
  (\ref{eq:inverse_comm_as_int}) admits a simple estimate leading to
  the desired inequality.
\end{proof}

\begin{rem*}
  Let us note that an estimate of this sort still exists when the
  spectra of $A$ and $B$ are interlaced provided
  $\dist(\Spec(A),\Spec(B))>0$. In the general case, as discussed in
  article \cite{bhatia_rosenthal97:_how_operator_eq_ax_xb_y}, it holds
  true that
  \[
  \Vert X\Vert\leq\frac{\pi}{2}
  \frac{\Vert Y\Vert}{\dist(\Spec(A),\Spec(B))}\,.
  \]
\end{rem*}

\begin{prop}
  \label{prop:trans_probability}
  Assume that the propagator $U(t,s)$ admits a Floquet decomposition
  (\ref{eq:Floquet-decomp}) which is continuously differentiable in
  the strong sense. Let $P(t,\cdot)$ be the projector-valued measure
  from the spectral decomposition of $H(t)$. Let
  $\Delta_{1},\Delta_{2}\subset\R$ be two intervals such that
  $\dist(\Delta_{1},\Delta_{2})>0$. Then it holds true that
  \begin{displaymath}
    \forall s,t\in\R,\textrm{~}\Vert P(t,\Delta_{1})
    U(t,s)P(s,\Delta_{2})\Vert
    \leq\frac{2\|S_F\|}{\dist(\Delta_{1},\Delta_{2})}\,.
  \end{displaymath}
  In particular, if $E_{n}(t)$ and $E_{m}(s)$ are eigenvalues of
  $H(t)$ and $H(s)$, respectively, $E_{n}(t)\neq{}E_{m}(s)$, and if
  $P_{n}(t)$ and $P_{m}(s)$ denote the projectors onto the
  corresponding eigenspaces then
  \[
  \Vert P_{n}(t)U(t,s)P_{m}(s)\Vert
  \leq\frac{2\|S_F\|}{|E_{n}(t)-E_{m}(s)|}\,.
  \]
\end{prop}

\begin{proof}
  Using relation (\ref{eq:H_eq_S_plus_HAsendw}) one verifies the
  equality
  \begin{equation}
    \label{eq:HU_UH}
    H(t)U(t,s)-U(t,s)H(s)
    = U(t,0)\big(e^{\ii tH_F}S_F(t)e^{-\ii tH_F}
    -e^{\ii sH_F}S_F(s)e^{-\ii sH_F}\big)U(0,s)
  \end{equation}
  which is valid on $\Dom(H(0))$. In particular, the LHS of
  (\ref{eq:HU_UH}) extends to an operator bounded on $\sH$ whose norm
  may be estimated from above by $2\|S_F\|$. Setting
  $A=H(t)P(t,\Delta_{1})$, $B=H(s)P(s,\Delta_{2})$ and
  $X=P(t,\Delta_{1})U(t,s)P(s,\Delta_{2})$, one easily finds that
  \begin{displaymath}
    AX-XB = P(t,\Delta_{1})(H(t)U(t,s)-U(t,s)H(s))P(s,\Delta_{2}).
  \end{displaymath}
  If the intervals $\Delta_{1}$, $\Delta_{2}$ are bounded then
  Lemma~\ref{lem:AX-XB_eq_Y} implies that
  $\|X\|\leq2\|S_F\|/\dist(\Delta_{1},\Delta_{2})$. If the intervals
  are not bounded one can use a limit procedure.
\end{proof}

\section{Extension: a higher order of differentiability of the Floquet
  decomposition}
\label{sec:higher_order}

Under assumptions on higher order differentiability in the strong
sense of the operator-valued function $U_F(t)$ in
(\ref{eq:Floquet-decomp}) one can extend the conclusions of
Proposition~\ref{prop:energy_bounded} and
Proposition~\ref{prop:trans_probability}. To this end, as an auxiliary
tool we first need to state some basic facts concerning the multiple
commutators.

\subsection{Multiple commutators}

\begin{define}
  \label{def:alpha_n}
  Let $A$ be a selfadjoint operator in $\sH$, $X\in\sB(\sH)$ and
  $n\in\Z_+$. The sesquilinear form
  \begin{displaymath}
    \alpha_n(\xi,\eta) = \sum_{k=0}^n \binom{n}{k}(-1)^k\,
    \langle XA^k\xi,A^{n-k}\eta\rangle_\sH
  \end{displaymath}
  is well defined on $\Dom(A^n)$. If it is bounded then there exists a
  unique bounded operator, denoted by $\ad_A^{\,n}X$, such that
  \begin{displaymath}
    \forall\xi,\eta\in\Dom(A^n),\textrm{~}
    \alpha_n(\xi,\eta) = \langle(\ad_A^{\,n}X)\xi,\eta\rangle_\sH.
  \end{displaymath}
  If this is the case we shall say that (the $n$-multiple commutator)
  $\ad_A^{\,n}X$ exists in $\sB(\sH)$.
\end{define}

\begin{rem*}
  Some elementary facts follow immediately from the definition.
  Suppose that $B=B^\ast$ is bounded. Then $\ad_B^{\,n}X\in\sB(\sH)$
  exists for all $n\in\Z_+$ and it holds
  \begin{displaymath}
    \ad_B^{\,n}X = \sum_{k=0}^n \binom{n}{k}(-1)^kB^{n-k}XB^k.
  \end{displaymath}
  Moreover, in this case $\ad_{A+B}X\in\sB(\sH)$ exists if and only if
  $\ad_AX\in\sB(\sH)$ exists and then $\ad_{A+B}X=\ad_AX+\ad_BX$.
\end{rem*}

\begin{define}
\label{def:CnA}
Suppose that $A=A^\ast$ in $\sH$. For every $n\in\Z_+$ we introduce
the linear subspace $\CC_n(A)\subset\sB(\sH)$ formed by those bounded
operators $X$ for which the commutators $\ad_A^{\,k}X\in\sB(\sH)$
exist for all $k=0,1,\ldots,n$.
\end{define}

\begin{rem*}
  Clearly, $\ad_A^{\,0}X=X$ and $\CC_0(A)=\sB(\sH)$. From the
  definition it is also obvious that the vector spaces are nested,
  i.e.,
  \begin{equation}
    \label{eq:Cn_nested}
    \CC_0(A)\supset\CC_1(A)\supset\CC_2(A)\supset\ldots.
  \end{equation}
\end{rem*}

\begin{lem}
  \label{lem:comm_Xast_XY}
  Suppose that $A=A^\ast$ and $X,Y\in\sB(\sH)$. If the commutators
  $\ad_AX,\ad_AY\in\sB(\sH)$ exist then there also exist
  $\ad_AX^\ast,\ad_A(XY)\in\sB(\sH)$ and it holds
  \renewcommand{\labelenumi}{(\roman{enumi})}
  \renewcommand{\theenumi}{\roman{enumi}}
  \begin{enumerate}
    \item\label{item:comm_1} $X(\Dom A)\subset\Dom A$,
    \item\label{item:comm_2} $\ad_AX^\ast=-(\ad_AX)^\ast$,
    \item\label{item:comm_3} $\ad_A(XY)=(\ad_AX)Y+X\ad_AY$.
  \end{enumerate}
\end{lem}

\begin{proof}
  To show (\ref{item:comm_1}) choose $\xi\in\Dom{}A$. By definition,
  for all $\eta\in\Dom{}A$ we have
  \begin{displaymath}
    \langle X\xi,A\eta\rangle =  \langle XA\xi,\eta\rangle
    + \langle(\ad_AX)\xi,\eta\rangle.
  \end{displaymath}
  Hence $X\xi$ belongs to $\Dom{}A^\ast=\Dom{}A$. Point
  (\ref{item:comm_2}) follows from the equality
  \begin{displaymath}
    \langle X^\ast\xi,A\eta\rangle-\langle X^\ast A\xi,\eta\rangle
    = -\overline{\langle X\eta,A\xi\rangle-\langle XA\eta,\xi\rangle}
    = \langle-(\ad_AX)^\ast\xi,\eta\rangle
  \end{displaymath}
  which is valid for all $\xi,\eta\in\Dom{}A$. For $\xi,\eta$ from the
  same domain we know, by points (\ref{item:comm_1}) and
  (\ref{item:comm_2}), that $Y\xi,X^\ast\eta\in\Dom{}A$. Thus we have
  the equality
  \begin{displaymath}
    \langle XY\xi,A\eta\rangle-\langle XYA\xi,\eta\rangle
    = \langle XY\xi,A\eta\rangle-\langle XAY\xi,\eta\rangle
    +\langle Y\xi,AX^\ast\eta\rangle-\langle YA\xi,X^\ast\eta\rangle
  \end{displaymath}
  with the RHS being equal to
  $\langle(\ad_AX)Y\xi,\eta\rangle+\langle(\ad_AY)\xi,X^\ast\eta\rangle$.
  Point (\ref{item:comm_3}) follows.
\end{proof}

\begin{rem*}
  Lemma~\ref{lem:comm_Xast_XY} implies that $\ad_AX\in\sB(\sH)$ exists
  if and only if $\Dom(A)$ is invariant with respect to $X$ and the
  operator $AX-XA$ is bounded on this domain. If this is the case then
  $\ad_AX=AX-XA$ on $\Dom(A)$.
\end{rem*}

\begin{lem}
  \label{lem:ad_AX_limit}
  Let $\{X_n\}_n$ be a sequence of bounded operators in $\sH$ such
  that the commutators $\ad_AX_n\in\sB(\sH)$ exist for all $n$. If the
  sequence $\{X_n\}_n$ converges weakly to a bounded operator $X$ and
  the sequence $\{\ad_AX_n\}_n$ converges weakly to a bounded operator
  $Y$ then $\ad_AX\in\sB(\sH)$ exists and equals $Y$.
\end{lem}

\begin{rem*}
  From Lemma~\ref{lem:ad_AX_limit} it follows that the linear operator
  $\ad_A$ on $\sB(\sH)$, with $\Dom(\ad_A)=\CC_1(A)$, is closed.
\end{rem*}

\begin{proof}
  Let $\xi,\eta\in\Dom(A)$ be arbitrary vectors. By definition, for
  all $n$,
  \begin{displaymath}
    \langle X_n\xi,A\eta\rangle-\langle X_nA\xi,\eta\rangle
    = \langle(\ad_AX_n)\xi,\eta\rangle.
  \end{displaymath}
  It suffices to send $n$ to infinity.
\end{proof}

\begin{prop}
  \label{prop:algsCn}
  The following statements are true for all $X\in\sB(\sH)$ and
  $n\in\Z_+$:
  \renewcommand{\labelenumi}{(\roman{enumi})}
  \renewcommand{\theenumi}{\roman{enumi}}
  \begin{enumerate}
  \item\label{item:Cn_1} If $X\in\CC_n(A)$ then
    $X(\Dom{}A^k)\subset\Dom{}A^k$ for all $k=0,1,\ldots,n$.
  \item\label{item:Cn_2} $X\in\CC_{n+1}(A)$ if and only if
    $\ad_AX\in\sB(\sH)$ exists and belongs to $\CC_n(A)$. Moreover, if
    this is the case then
    \begin{displaymath}
      \ad_A^{\,k}(\ad_AX) = \ad_A^{\,k+1}X\textrm{~~for~}
      k=0,1,\ldots,n.
    \end{displaymath}
  \item\label{item:Cn_3} $\CC_n(A)$ is a $\ast$-subalgebra of
    $\sB(\sH)$.
  \end{enumerate}
\end{prop}

\begin{proof}
  (i) We shall show that, for a given $k\in\Z_+$, the domain
  $\Dom{}A^k$ is invariant with respect to all $X\in\CC_n(A)$ as long
  as $n\geq{}k$. Recalling that the spaces $\CC_n(A)$ are nested it
  suffices to consider the case of $n=k$. To this end, we shall
  proceed by induction in $k$.  For $k=0$ the statement is trivial.
  Suppose that the statement holds true for all $\ell$,
  $0\leq\ell\leq{}k$. Choose $X\in\CC_{k+1}(A)$. The induction
  hypothesis implies that for any $\xi\in\Dom(A^{k+1})$ and
  $1\leq\ell\leq{}k+1$, $XA^\ell\xi\in\Dom(A^{k+1-\ell})$. By the
  definition of $\ad_A^{\,k+1}X$ we have the equality
  \begin{displaymath}
    \langle X\xi,A^{k+1}\eta\rangle
    = \sum_{\ell=1}^{k+1}\binom{k+1}{\ell}(-1)^{\ell+1}
    \langle A^{k+1-\ell}XA^\ell\xi,\eta\rangle
    +\langle(\ad_A^{\,k+1})X\xi,\eta\rangle,
  \end{displaymath}
  valid for all $\eta\in\Dom(A^{k+1})$. Hence $X\xi\in\Dom(A^{k+1})$.

  (ii) By the very definition, if $X\in\CC_{n+1}(A)$ then
  $\ad_AX\in\sB(\sH)$ exists. If $0\leq{}m\leq{}n$ and
  $\xi,\eta\in\Dom(A^{m+1})$ then simple algebraic manipulations lead
  to the equality
  \begin{equation}
    \label{eq:adAmXplus1}
    \sum_{k=0}^m\binom{m}{k}(-1)^k
    \langle(\ad_AX)A^k\xi,A^{m-k}\eta\rangle
    = \langle(\ad_A^{\,m+1}X)\xi,\eta\rangle.
  \end{equation}
  The both sides in (\ref{eq:adAmXplus1}) extend in a unique way to
  the domain $\xi,\eta\in\Dom(A^m)$. It follows that
  $\ad_A^{\,m}(\ad_AX)\in\sB(\sH)$ exists and equals $\ad_A^{\,m+1}X$.
  Hence $\ad_AX\in\CC_n(A)$. Conversely, suppose that
  $\ad_AX\in\CC_n(A)$. For any $m$, $0\leq{}m\leq{}n$, and
  $\xi,\eta\in\Dom(A^{m+1})$, one finds, again with the aid of simple
  algebraic manipulations, that
  \begin{displaymath}
    \langle\ad_A^{\,m}(\ad_AX)\xi,\eta\rangle
    = \sum_{k=0}^{m+1}\binom{m+1}{k}(-1)^k
    \langle XA^k\xi,A^{m+1-k}\eta\rangle.
  \end{displaymath}
  Hence $\ad_A^{\,m+1}X\in\sB(\sH)$ exists and thus
  $X\in\CC_{n+1}(A)$.

  (iii) First let us show that $X^\ast\in\CC_n(A)$ provided the same
  is true for $X$. We shall proceed by induction in $n$. The case
  $n=0$ is obvious. Suppose that the claim is true for $n$. If
  $X\in\CC_{n+1}(A)$ then, by the already proved point
  (\ref{item:Cn_2}) of the current proposition, $\ad_AX\in\CC_n(A)$.
  By the induction hypothesis and Lemma~\ref{lem:comm_Xast_XY}
  ad~(\ref{item:comm_2}) we have
  $\ad_AX^\ast=-(\ad_AX)^\ast\in\CC_n(A)$. Referring once more to point
  (\ref{item:Cn_2}) of the current proposition we conclude that indeed
  $X^\ast\in\CC_{n+1}(A)$.

  Finally let us show that $XY\in\CC_n(A)$ provided $X,Y\in\CC_n(A)$.
  We shall proceed by induction in $n$. The case $n=0$ is again
  obvious. Suppose that the claim is true for $n$. If
  $X,Y\in\CC_{n+1}(A)$ then, by point (\ref{item:Cn_2}) of the current
  proposition, $\ad_AX,\ad_AY\in\CC_n(A)$. By the induction hypothesis
  and Lemma~\ref{lem:comm_Xast_XY} ad~(\ref{item:comm_3}) we have
  \begin{displaymath}
    \ad_A(XY)=(\ad_AX)Y+X\ad_AY\in\CC_n(A).
  \end{displaymath}
  Referring again to point (\ref{item:Cn_2}) of the current
  proposition we conclude that $XY\in\CC_{n+1}(A)$.
\end{proof}

\begin{rem*}
  As an immediate consequence of Proposition~\ref{prop:algsCn} ad
  (\ref{item:Cn_1}) it holds $\Dom(A^k)=\Dom((A+B)^k)$ for
  $k=0,1,\ldots,p$, provided $B\in\CC_{p-1}(A)$ for some $p\in\N$.
\end{rem*}

\begin{define}
  \label{def:Cn_uniform}
  Let $A$ be a self-adjoint operator on $\sH$ and $X(t)\in\sB(\sH)$ be
  an operator-valued function, with the variable $t$ running over
  $\R$, and let $n\in\Z_+$. We shall say that $X(t)$ is in the algebra
  $\CC_n(A)$ uniformly if $X(t)\in\CC_n(A)$ for all $t\in\R$ and
  \begin{displaymath}
    \sup_{t\in\R}\sum_{k=0}^n\|\ad_{A}^{\,k}X(t)\| < \infty.
  \end{displaymath}
\end{define}

\begin{rems*}
  Of course, the operator-valued function $X(t)$ may be constant.

  From Proposition~\ref{prop:algsCn} ad~(\ref{item:Cn_2}) one
  immediately deduces that an operator-valued function
  $X(t)\in\sB(\sH)$ is in $\CC_{n+1}(A)$ uniformly if and only if
  $X(t)$ is uniformly bounded and $\ad_{A}X(t)$ is in $\CC_n(A)$
  uniformly. Moreover, a straightforward induction procedure based on
  this observation jointly with Lemma~\ref{lem:comm_Xast_XY}
  ad~(\ref{item:comm_2}) and ad~(\ref{item:comm_3}) implies that if
  $X(t)$ and $Y(t)$ are in $\CC_n(A)$ uniformly then also
  $X(t)^\ast$ and $X(t)Y(t)$ are in $\CC_n(A)$ uniformly.
\end{rems*}

\begin{lem}
  \label{lem:CnAplusB_uniform}
  Let $A$ be a self-adjoint operator and $B\in\CC_{p-1}(A)$ be a
  Hermitian operator for some $p\in\N$. Then an operator-valued
  function $X(t)\in\sB(\sH)$, with $t\in\R$, is in $\CC_p(A)$
  uniformly if and only if $X(t)$ is in $\CC_p(A+B)$ uniformly.
\end{lem}

\begin{proof}
  Clearly it suffices to prove only one implication since the other
  one follows after replacing $A$ by $A+B$ and $B$ by $-B$ (while
  making use of the simple fact that $\ad_{A+B}^{\,m}B=\ad_A^{\,m}B$).
  We shall proceed by induction in $p$.

  As far as the case $p=1$ is concerned we assume that $B\in\CC_0(A)$
  and $X(t)$ is in $\CC_1(A)$ uniformly. This in particular means that
  $X(t)$ is a uniformly bounded operator-valued function and hence the
  same is true for $\ad_{B}X(t)=BX(t)-X(t)B$. Now it suffices to take
  into account the equality $\ad_{A+B}X(t)=\ad_{A}X(t)+\ad_{B}X(t)$.

  Let us now assume that the lemma has been proved for some $p\in\N$,
  and that $B\in\CC_p(A)$ and $X(t)$ is in $\CC_{p+1}(A)$ uniformly.
  Now we can repeatedly apply the remarks following
  Definition~\ref{def:Cn_uniform}. Firstly, $\ad_{B}X(t)=BX(t)-X(t)B$
  is in $\CC_p(A)$ uniformly.  Secondly, $X(t)$ is uniformly bounded
  and $\ad_{A}X(t)$ is in $\CC_p(A)$ uniformly. Consequently,
  $\ad_{A+B}X(t)$ is in $\CC_p(A)$ uniformly as well. By the induction
  hypothesis, $\ad_{A+B}X(t)$ is in $\CC_p(A+B)$ uniformly. This in
  turn implies that $X(t)$ is in $\CC_{p+1}(A+B)$ uniformly.
\end{proof}

In the particular case when the operator-valued function $X(t)$ is
constant Lemma~\ref{lem:CnAplusB_uniform} reduces to the following
statement.

\begin{lem}
  \label{lem:CnAplusB_const}
  Let $A$ be a self-adjoint operator on $\sH$ and $B\in\CC_{p-1}(A)$
  for some $p\in\N$, and suppose that $B=B^\ast$. Then
  $\CC_k(A)=\CC_k(A+B)$ for $k=0,1,\ldots,p$.
\end{lem}

We shall also need the following algebraic lemma.

\begin{lem}
  \label{lem:HpU_commute}
  Suppose that $A=A^\ast$ and $B\in\CC_p(A)$ for some $p\in\Z_+$. Then
  the following claims are true:
  \renewcommand{\labelenumi}{(\roman{enumi})}
  \renewcommand{\theenumi}{\roman{enumi}}
  \begin{enumerate}
  \item\label{item:ApB} On $\Dom(A^p)$ it holds
    \begin{equation}
      \label{eq:ApB}
      A^pB = \sum_{k=0}^p\binom{p}{k}\left(\ad_A^{\,p-k}B\right)A^k.
    \end{equation}
  \item\label{item:AplusBp} There exist polynomials
    $F_{p,k}(\gx_0,\gx_1,\ldots,\gx_{p-k-1})$, $k=0,1,\ldots,p-1$, in
    non-commutative variables $\gx_j$, with non-negative integer
    coefficients and such that it holds
    \begin{equation}
      \label{eq:AplusBp}
      (A+B)^p = A^p+\sum_{k=0}^{p-1}
      F_{p,k}\big(B,\ad_AB,\ldots,\ad_A^{\,p-k-1}B\big)A^k
    \end{equation}
    on $\Dom(A^p)$.
  \end{enumerate}
\end{lem}

\begin{proof}
  (i) By Proposition~\ref{prop:algsCn} ad~(\ref{item:Cn_1}) the both
  sides of (\ref{eq:ApB}) are well defined on $\Dom(A^p)$. To verify
  (\ref{eq:ApB}) one can proceed by induction in $p$ which amounts to
  simple algebraic manipulations. We omit the details.

  (ii) Again, the both sides of (\ref{eq:AplusBp}) are well defined on
  $\Dom(A^p)$. One can proceed by induction in $p$. Set, by
  convention, $F_{p,p}=1$. To carry out the induction step let us
  write
  \begin{displaymath}
    (A+B)^{p+1} = \sum_{k=0}^{p}
    F_{p,k}\big(B,\ad_AB,\ldots,\ad_A^{\,p-k-1}B\big)A^k(A+B)
  \end{displaymath}
  and apply claim~(\ref{item:ApB}) of the current lemma to manage the
  term $A^kB$ on the RHS. By comparison one arrives at the recursion
  rule
  \begin{displaymath}
    F_{p+1,k}(\gx_0,\gx_1,\ldots,\gx_{p-k})
    = F_{p,k-1}(\gx_0,\gx_1,\ldots,\gx_{p-k})
    +\sum_{\ell=k}^p\binom{\ell}{k}
    F_{p,\ell}(\gx_0,\gx_1,\ldots,\gx_{p-\ell-1})\,\gx_{\ell-k}
  \end{displaymath}
  from which claim~(\ref{item:AplusBp}) easily follows.
\end{proof}

\subsection{Differentiable Floquet decompositions}

In this section we shall assume that
\begin{displaymath}
  V(t):=H(t)-H(0)
\end{displaymath}
is a uniformly bounded operator-valued function. Of course, it is
Hermitian and $T$-periodic. We also assume that we are given a Floquet
decomposition (\ref{eq:Floquet-decomp}) of the corresponding
propagator $U(t,s)$. For $p\in\Z_+$ let us set
\begin{displaymath}
  \AA_p^0 = \CC_p(H_0),\textrm{~}\AA_p = \CC_p(H_F).
\end{displaymath}
Here and everywhere in this section we write shortly $H_0=H(0)$ and
$S_0=S_F(0)$. Thus we have $H_0=H_F+S_0$ (see (\ref{eq:S_eq_diff_UF})
and (\ref{eq:H_eq_S_plus_HAsendw})).

If the Floquet decomposition is continuously differentiable in the
strong sense and $S_0\in\AA_{p-1}$ for some $p\in\N$ then
Lemma~\ref{lem:CnAplusB_const} tells us that $\AA_k=\AA_k^0$ for
$k=0,1,\ldots,p$.

\begin{lem}
  \label{lem:Floquet_diffp_VA_UA}
  Let us assume that $p\in\N$ and $V(t)\in{}C^{p-1}(\R)$ in the strong
  sense, and that the propagator $U(t,s)$ admits a Floquet
  decomposition (\ref{eq:Floquet-decomp}) which is $p$ times
  continuously differentiable in the strong sense. If
  \begin{equation}
    \label{eq:V_in_A_uni}
    V^{(k)}(t)\textrm{~is in~}\AA_{p-1-k}\textrm{~uniformly for~}
    k=0,1,\dots,p-1,
  \end{equation}
  then
  \begin{equation}
    \label{eq:UF_in_A_uni}
    U_F^{(k)}(t)\textrm{~is in~}\AA_{p-k}\textrm{~uniformly for~}
    k=0,1,\ldots,p.
  \end{equation}
  Moreover, $\AA_k=\AA_k^0$ for $k=0,1,\ldots,p$, and
  $S_F(t)=\ii\,U_F(t)^{-1}\partial_tU_F(t)$ is in $\AA_{p-1}$
  uniformly.
\end{lem}

\begin{proof}
  For the proof we shall need the relation
  \begin{equation}
    \label{eq:comm_H_FU_F}
    \ad_{H_F}U_F(t) = U_F(t)S_F(t)-\big(S_0+V(t)\big)U_F(t).
  \end{equation}
  Here $U_F(t)$ preserves the domain $\Dom(H_F)=\Dom(H_0)$. Equality
  (\ref{eq:comm_H_FU_F}) follows from (\ref{eq:H_eq_S_plus_HAsendw})
  and the substitution
  \begin{displaymath}
    H(t) = H_0+V(t) = S_0+H_F+V(t).
  \end{displaymath}
  From the differentiability of $U_F(t)$ it follows that $S_F(t)$
  belongs to $C^{p-1}(\R)$ in the strong sense. Thus all derivatives
  of $S_F(t)$ up to the order $p-1$ are uniformly bounded (due to the
  periodicity). With the aid of Lemma~\ref{lem:ad_AX_limit} we derive
  from (\ref{eq:comm_H_FU_F}) that
  \begin{equation}
    \label{eq:comm_H_FdiffU_F}
    \ad_{H_F}U_F^{(k)}(t)
    = \frac{\dd^k}{\dd t^k}
    \left(U_F(t)S_F(t)-\big(S_0+V(t)\big)U_F(t)\right)\in\sB(\sH)
    \textrm{~~for~}k=0,1,\ldots,p-1,
  \end{equation}
  (with all derivatives taken in the strong sense). Moreover,
  $\ad_{H_F}U_F^{(k)}(t)$ is uniformly bounded for
  $0\leq{}k\leq{}p-1$. Note also that (\ref{eq:V_in_A_uni}) can be
  rewritten in the form
  \begin{equation}
    \label{eq:V_in_A_kl}
    V^{(k)}(t)\textrm{~is in~}\AA_{p-1-\ell}\textrm{~uniformly if~}
    0 \leq k \leq \ell \leq p-1
  \end{equation}
  since the algebras $\AA_r$ are nested,
  $\AA_0\supset\AA_1\supset\AA_2\supset\ldots$, (see
  (\ref{eq:Cn_nested})).

  We shall verify that, for $\ell=0,1,\ldots,p$,
  \begin{equation}
    \label{eq:UF_in_A_kl}
    U_F^{(k)}(t)\textrm{~is in~}\AA_{p-\ell}\textrm{~uniformly if~}
    0\leq k\leq\ell\leq p.
  \end{equation}
  Since
  \begin{displaymath}
    S_F^{(k)}(t) = \frac{\dd^k}{\dd t^k}
    \left(\ii\,U_F(t)^{\ast}U_F'(t)\right)
  \end{displaymath}
  and $\AA_n$ is a $\ast$-algebra relation (\ref{eq:UF_in_A_kl})
  implies that, for $\ell=1,\ldots,p$,
  \begin{equation}
    \label{eq:S_in_A_kl}
    S_F^{(k)}(t)\textrm{~is in~}\AA_{p-\ell}\textrm{~uniformly if~}
    0\leq k\leq\ell-1\leq p-1.
  \end{equation}
  To show (\ref{eq:UF_in_A_kl}) we shall proceed by a finite
  descending induction in $\ell$. According to the assumptions of the
  lemma, $U_F^{(k)}(t)$ is uniformly bounded for $0\leq{}k\leq{}p$ and
  the case $\ell=p$ follows. Suppose now that (\ref{eq:UF_in_A_kl}) is
  valid for some $\ell$, $1\leq\ell\leq{}p$. Then for the same $\ell$,
  (\ref{eq:S_in_A_kl}) is valid as well. Moreover, replacing $\ell$ by
  $\ell-1$ in (\ref{eq:V_in_A_kl}) one knows that $V^{(k)}(t)$ is in
  $\AA_{p-\ell}$ uniformly for $0\leq{}k\leq\ell-1$. Thus if
  $0\leq{}k\leq\ell-1$ then from the fact that $\AA_{p-\ell}$ is an
  algebra and from the induction hypothesis one deduces that the RHS
  of (\ref{eq:comm_H_FdiffU_F}) is in $\AA_{p-\ell}$ uniformly. This
  implies, in virtue of Proposition~\ref{prop:algsCn}
  ad~(\ref{item:Cn_2}), that $U_F^{(k)}(t)$ is in $\AA_{p-\ell+1}$
  uniformly. This completes the induction step and relation
  (\ref{eq:UF_in_A_kl}) is verified.

  Setting $k=\ell$ in (\ref{eq:UF_in_A_kl}) one obtains
  (\ref{eq:UF_in_A_uni}). Setting $k=0$ and $\ell=1$ in
  (\ref{eq:S_in_A_kl}) one finds that $S_F(t)$ is $\AA_{p-1}$
  uniformly. In particular, $S_0\equiv{}S_F(0)$ belongs to
  $\AA_{p-1}=\CC_{p-1}(H_F)$. Since $H_0=H_F+S_0$ from
  Lemma~\ref{lem:CnAplusB_const} we know that
  \begin{displaymath}
    \AA_k = \CC_k(H_F) = \CC_k(H_0) = \AA_k^0
  \end{displaymath}
  for $k=0,1,\ldots,p$.
\end{proof}

\begin{cor}
  \label{cor:Floquet_diffp_VA0_UA}
  Lemma~\ref{lem:Floquet_diffp_VA_UA} remains still true if
  $\AA_{p-1-k}$ is replaced by $\AA^0_{p-1-k}$ in the condition
  (\ref{eq:V_in_A_uni}).
\end{cor}

\begin{proof}
  We shall proceed by induction in $p$. For $p=1$ we have
  $\AA_0^0=\AA_0=\sB(\sH)$ and thus replacing $\AA_{p-1-k}$ by
  $\AA^0_{p-1-k}$ in (\ref{eq:V_in_A_uni}) does not mean any change.
  Let us now suppose that the claim is true for some $p\in\N$. And we
  assume that $V^{(k)}(t)$ is in $\AA^0_{p-k}$ uniformly for
  $k=0,1,\ldots,p$. Of course, the other assumptions of
  Lemma~\ref{lem:Floquet_diffp_VA_UA}, except of the condition
  (\ref{eq:V_in_A_uni}), are satisfied as well, namely
  $V(t)\in{}C^{p}(\R)$ in the strong sense, and the propagator
  $U(t,s)$ admits a Floquet decomposition (\ref{eq:Floquet-decomp})
  which is $p+1$ times continuously differentiable in the strong
  sense. Since $\AA^0_{p-k}\subset\AA^0_{p-1-k}$, $V^{(k)}(t)$ is in
  $\AA^0_{p-1-k}$ uniformly for $k=0,1,\ldots,p-1$. By the induction
  hypothesis, Lemma~\ref{lem:Floquet_diffp_VA_UA} is applicable for
  the value $p$ and therefore, in particular, $\AA_k=\AA_k^0$ for
  $k=0,1,\ldots,p$. Hence $V^{(k)}(t)$ is in $\AA_{p-k}$ uniformly for
  $k=0,1,\ldots,p$, which is nothing but condition
  (\ref{eq:V_in_A_uni}) with $p$ being replaced by $p+1$. It follows
  that the conclusions of Lemma~\ref{lem:Floquet_diffp_VA_UA} hold
  true for the value $p+1$ as well.
\end{proof}

\begin{prop}
  \label{prop:Floquet_diffp_HpU}
  Let us assume that $p\in\N$ and $V(t)\in{}C^{p-1}(\R)$ in the strong
  sense, and that the propagator $U(t,s)$ admits a Floquet
  decomposition (\ref{eq:Floquet-decomp}) which is $p$ times
  continuously differentiable in the strong sense. If
  \begin{equation}
    \label{eq:Vk_in_A0_uni}
    V^{(k)}(t)\textrm{~is in~}\AA_{p-1-k}^0\textrm{~uniformly for~}
    k=0,1,\dots,p-1,
  \end{equation}
  then $U(t,0)$, $t\in\R$, preserves the domain $\Dom(H_0^{\,p})$ and
  \begin{displaymath}
    \forall\psi\in\Dom(H_0^{\,p}),\textrm{~}
    \sup_{t\in\R}\|H(t)^pU(t,0)\psi\| < \infty.
  \end{displaymath}
\end{prop}

\begin{proof}
  From Corollary~\ref{cor:Floquet_diffp_VA0_UA} we know that $U_F(t)$
  is in $\AA_p$ uniformly and $S_F(t)$ is in $\AA_{p-1}$ uniformly.
  Since $S_0\in\AA_{p-1}$ and $H_0=H_F+S_0$,
  Lemma~\ref{lem:CnAplusB_const} tells us that $\AA_k=\AA_k^0$ for
  $0\leq{}k\leq{}p$. From the relations $V(t)\in\AA_{p-1}^0$ and
  $S_0\in\AA_{p-1}$ it also follows that
  \begin{equation}
    \label{eq:Dom_H0k_Htk_HFk}
    \Dom(H_0^{\,k}) = \Dom(H(t)^k) = \Dom(H_F^{\,k})
    \textrm{~~for~}k=0,1,\ldots,p,
  \end{equation}
  see Proposition~\ref{prop:algsCn} ad~(\ref{item:Cn_1}).
  Furthermore, from the Floquet decomposition
  (\ref{eq:Floquet-decomp}) and the above observation on $U_F(t)$ one
  deduces that $U(t,0)$ is in $\AA_p=\AA_p^0$ uniformly and therefore
  $U(t,0)(\Dom(H_0^{\,p}))\subset\Dom(H_0^{\,p})$.

  Suppose that $\psi\in\Dom(H_0^{\,p})$. From
  (\ref{eq:Floquet-decomp}) and (\ref{eq:H_eq_S_plus_HAsendw}) one
  finds that
  \begin{displaymath}
    H(t)^{p}U(t,0)\psi
    = U_F(t)\big(H_F+S_F(t)\big)^pe^{-\ii tH_F}\psi.
  \end{displaymath}
  With the aid of equality~(\ref{eq:AplusBp}) of
  Lemma~\ref{lem:HpU_commute} ad~(\ref{item:AplusBp}) one derives the
  estimate
  \begin{displaymath}
    \|H(t)^{p}U(t,0)\psi\| \leq \sum_{k=0}^{p}
    F_{p,k}(\gS_0,\gS_1,\ldots,\gS_{p-k-1})\,\|H_F^{\,k}\psi\|
  \end{displaymath}
  where
  \begin{displaymath}
    \gS_k := \sup_{t\in\R}\|\ad_{H_F}^{\,k}S_F(t)\|,\textrm{~}
    k=0,1,\ldots,p-1.
  \end{displaymath}
  The proposition follows.
\end{proof}

\begin{lem}
  \label{lem:Xn_in_Ap}
  Under the same assumptions as in
  Proposition~\ref{prop:Floquet_diffp_HpU}, the operators
  \begin{equation}
    \label{eq:Xn_def}
    X_n(t,s) = \sum_{k=0}^n\binom{n}{k}(-1)^k
    H(t)^{n-k}U(t,s)H(s)^k,\textrm{~~}n=0,1,\ldots,p,
  \end{equation}
  are well defined on $\Dom(H_0^{\,n})$. Moreover, $X_n(t,s)$ extends
  in a unique way to a bounded operator on $\sH$ which is in
  $\AA_{p-n}$ uniformly with respect to the variables $(t,s)\in\R^2$.
\end{lem}

\begin{proof}
  In the same way as in the proof of
  Proposition~\ref{prop:Floquet_diffp_HpU}, we deduce from the
  assumptions that equalities (\ref{eq:Dom_H0k_Htk_HFk}) hold true as
  well as that $U_F(t)$ is in $\AA_p=\AA_p^0$ uniformly and $S_F(t)$
  is in $\AA_{p-1}$ uniformly. Moreover, Proposition~\ref{prop:algsCn}
  ad~(\ref{item:comm_1}) tells us that $U_F(t)$ preserves
  $\Dom(H_F^{\,k})$ for $k=0,1,\ldots,p$.

  From (\ref{eq:Floquet-decomp}) and (\ref{eq:H_eq_S_plus_HAsendw}) it
  follows that
  \begin{displaymath}
    X_n(t,s) = U_F(t)Z_n(t,s)U_F(s)^{\,-1}
  \end{displaymath}
  where
  \begin{displaymath}
    Z_n(t,s) = \sum_{k=0}^n\binom{n}{k}(-1)^k
    \big(H_F+S_F(t)\big)^{n-k}e^{-\ii(t-s)H_F}\big(H_F+S_F(s)\big)^k.
  \end{displaymath}
  It suffices to show that $Z_n(t,s)$ is well defined on
  $\Dom(H_F^{\,n})$ and extends to a bounded operator on $\sH$ which
  is in $\AA_{p-n}$ uniformly. To verify it we proceed by induction in
  $n$.

  For $n=0$, $Z_0(t,s)=e^{-\ii(t-s)H_F}$ fulfills
  $\ad_{H_F}^{\,k}Z_0(t,s)=0$ for all $k\ge1$ and so it is in $\AA_p$
  uniformly. To carry out the induction step observe that
  \begin{eqnarray*}
    Z_{n+1}(t,s)
    &=& \big(H_F+S_F(t)\big)Z_n(t,s)-Z_n(t,s)\big(H_F+S_F(s)\big) \\
    &=& \ad_{H_F}Z_n(t,s)+S_F(t)Z_n(t,s)-Z_n(t,s)S_F(s).
  \end{eqnarray*}
  The induction hypothesis and Proposition~\ref{prop:algsCn}
  ad~(\ref{item:Cn_2}) (see also Remarks following
  Definition~\ref{def:Cn_uniform}) imply that $\ad_{H_F}Z_n(t,s)$ is
  in $\AA_{p-n-1}$ uniformly. Recalling Proposition~\ref{prop:algsCn}
  ad~(\ref{item:Cn_3}) it also holds true that $S_F(t)Z_n(t,s)$ and
  $Z_n(t,s)S_F(s)$are in $\AA_{p-n-1}$ uniformly. This verifies the
  induction step and concludes the proof of the lemma.
\end{proof}

\begin{prop}
  \label{prop:Floquet_diffp_transprob}
  Under the same assumptions as in
  Proposition~\ref{prop:Floquet_diffp_HpU} (including condition
  (\ref{eq:Vk_in_A0_uni})), let $P(t,\cdot)$ be the projection-valued
  measure from the spectral decomposition of $H(t)$. Then there exists
  a constant $C_p\geq0$ such that for any couple of intervals
  $\Delta_1,\Delta_2\subset\R$ whose distance
  $\dist(\Delta_1,\Delta_2)$ is positive it holds true
  \begin{equation}
    \label{eq:trans_prob_higher_ord}
    \forall s,t\in\R,\textrm{~}
    \Vert P(t,\Delta_{1})U(t,s)P(s,\Delta_{2})\Vert
    \leq\frac{C_p}{{\dist(\Delta_{1},\Delta_{2})}^p}\,.
  \end{equation}
\end{prop}

\begin{proof}
  It suffices to verify the assertion for bounded intervals. The
  general case then follows by a limit procedure. Set
  $Y_n(t,s)=P(t,\Delta_{1})X_n(t,s)P(s,\Delta_{2})$ where $X_n(t,s)$
  is defined in (\ref{eq:Xn_def}), and $Q_1(t)=H(t)P(t,\Delta_{1})$,
  $Q_2(s)=H(s)P(s,\Delta_2)$. In particular,
  $Y_0(t,s)=P(t,\Delta_{1})U(t,s)P(s,\Delta_{2})$. From
  Lemma~\ref{lem:Xn_in_Ap} we know that the operator-valued functions
  $X_n(t,s)$ are uniformly bounded. If $0\leq{}n<p$ then it holds
  \begin{displaymath}
    Q_1(t)Y_n(t,s)-Y_n(t,s)Q_2(s) = Y_{n+1}(t,s).
  \end{displaymath}
  By Lemma~\ref{lem:AX-XB_eq_Y} we have the estimate
  \begin{displaymath}
    \|Y_n(t,s)\| \leq \frac{\|Y_{n+1}(t,s)\|}
    {\dist\big(\Spec(Q_1(t)),\Spec(Q_2(s))\big)}
    \leq \frac{\|Y_{n+1}(t,s)\|}{\dist(\Delta_1,\Delta_2)}\,.
  \end{displaymath}
  Applying this estimate consecutively for $n=0,1,\ldots,p-1$, we find
  that (\ref{eq:trans_prob_higher_ord}) holds true with
  $C_p=\sup_{(t,s)\in\R^2}\|X_p(t,s)\|$.
\end{proof}

\section{A solvable example: the time-dependent harmonic oscillator}

Let us consider the time-dependent harmonic oscillator
\begin{displaymath}
  H(t) = H_\omega+f(t)x,\textrm{~}H_{\omega }
  = -\frac{1}{2}\,\partial^{\,2}_{x}+\frac{\omega^2 x^2}{2}\,,
\end{displaymath}
in $\sH=L^2(\R,\dd{}x)$ where the function $f(t)$ is supposed to be
continuous and $T$ periodic. The Hamiltonians quadratic in $x$ and $p$
turn out to be quite attractive in various situations since they allow
for explicit computations. For example, a classical result is a
formula for the Green function computed in the framework of the
Feynman path integral \cite{feynman_hibbs65}, see also
\cite{schulman81} and comments on the literature therein. For purposes
of the present paper we need some of the results derived in
\cite{enss_veselic83:_bound_propag_states_time_dep_hamil} and
concerned with the dynamical properties of $H(t)$, see also an
additional analysis in \cite[Chp.~5]{bunimovich_etql91}. Let us also
mention that in \cite{hagedorn_etal86} it has been shown that the
Floquet operator associated to a time-dependent quadratic Hamiltonian
can only have either a pure point spectrum or a purely absolutely
transient continuous spectrum.

As pointed out in
\cite{enss_veselic83:_bound_propag_states_time_dep_hamil}, it holds
\begin{eqnarray*}
  U(t,0)^{-1}xU(t,0) &=&
  x\cos(\omega{}t)+\frac{p}{\omega}\sin(\omega{}t)
  -\frac{1}{\omega}\,\varphi_2(t,0), \\
  U(t,0)^{-1}pU(t,0) &=&
  -\,\omega x\sin(\omega{}t)+p\cos(\omega{}t)+\varphi_1(t,0),
\end{eqnarray*}
where the functions $\varphi_1(t,s)$ and $\varphi_2(t,s)$ are given in
(\ref{eq:varphi12}). Assume for a moment that $\varphi_1(t,0)$ and
$\varphi_2(t,0)$ are uniformly bounded. Under this assumption it is
obvious that if an initial condition $\psi$ belongs to the Schwartz
space $\sS$ then the quantity
$\langle{}U(t,0)\psi,P(p,x)U(t,0)\psi\rangle$ is uniformly bounded in
time for any polynomial $P(p,x)$ in the non-commuting variables
$p=-\ii\partial_x$ and $x$. In particular, for such an initial
condition, the mean value of energy is bounded uniformly. As stated in
\cite[Proposition~4.1]{enss_veselic83:_bound_propag_states_time_dep_hamil},
it follows that all trajectories $\{U(t,0)\psi;\,t\in\R\}$, for any
initial condition $\psi\in\sH$, are precompact subsets in $\sH$. This
in turn implies that the spectrum of the monodromy operator $U(T,0)$
is pure point (see Theorem~2.3 in
\cite{enss_veselic83:_bound_propag_states_time_dep_hamil}). The fact
that the mean value of energy is bounded for all initial conditions
from a total set has also the following consequence (see Lemma~3.3 in
\cite{enss_veselic83:_bound_propag_states_time_dep_hamil}):
\begin{displaymath}
  \forall\psi\in\sH,\textrm{~}
  \lim_{R\to\infty}\sup_{t\in\R}\,\|F(H_\omega>R)U(t,0)\psi\|=0
\end{displaymath}
where the symbol $F$ stands for the projection-valued measure from the
spectral decomposition of the operator indicated in the argument and
taken for a subset of the real line which is indicated in the argument
as well.

Let us note that paper
\cite{enss_veselic83:_bound_propag_states_time_dep_hamil} has finally
focused on the particular case $f(t)=\sin(2\pi{}t/T)$. In that case a
simple computation shows that the functions $\varphi_1(t,0)$ and
$\varphi_2(t,0)$ are bounded if and only if $2\pi/T\neq\omega$.

Let us now examine how Proposition~\ref{prop:energy_bounded} can be
applied to this example. We consider the non-resonant case
\begin{displaymath}
  T \notin \frac{2\pi}{\omega}\,\N.
\end{displaymath}
Let us write
\begin{displaymath}
  T = \frac{2\pi}{\omega}\,N+\Delta,\textrm{~with~}
  N\in\Z_+,\textrm{~}\Delta\in                      
  \Big]0,\frac{2\pi}{\omega}\Big[\,.                   
\end{displaymath}
As a first step one has to make a choice of a self-adjoint operator
$H_F$ so that $U(T,0)=\exp(-\ii{}TH_F)$. According to
Proposition~\ref{thm:propag_ultimite}, the monodromy operator
corresponding to $H(t)$ can be expressed in the form
\begin{equation}
  \label{eq:UT_eq_exp}
  U(T,0) = (-1)^N\exp\!\left(-\ii\Delta H_{\omega}
    +\ii\,\frac{\mu(T,0)}{\omega}\,p
    +\ii\,\nu(T,0)x+\ii\sigma(T,0)\right)
\end{equation}
where the functions $\mu(t,s)$ and $\nu(t,s)$ are given in
(\ref{eq:mu_nu}) and $\sigma(t,s)$ is given in (\ref{eq:sigma}).

We shall seek $H_F$ in the form
\begin{displaymath}
  H_F = H_{\omega}-\frac{\alpha}{\omega T}\,p
  -\frac{\beta}{T}\,x+\frac{\gamma}{T}
\end{displaymath}
for some $\alpha,\beta,\gamma\in\R$. Then it holds
\begin{eqnarray*}
  \exp(-\ii TH_F) &=& e^{-\ii\gamma}
  \exp\!\left(-\ii TH_\omega
    +\ii\,\frac{\alpha}{\omega}\,p+\ii\beta x\right) \\
  &=& \exp\!\left(-\ii\gamma
    +\ii\,\frac{\alpha^2+\beta^2}{2\omega^2T}\right)
  \exp\!\left(\ii\,\frac{\alpha}{\omega T}\,x\right)
  \exp\!\left(-\ii\,\frac{\beta}{\omega^2T}\,p\right)
  \exp\!\left(-\ii TH_\omega\right) \\
  && \times\,\exp\!\left(\ii\,\frac{\beta}{\omega^2T}\,p\right)
  \exp\!\left(-\ii\,\frac{\alpha}{\omega T}\,x\right).
\end{eqnarray*}
Here we have used that
\begin{equation}
  \label{eq:expsHexps}
  e^{\ii sx}H_{\omega } e^{-\ii sx}
  = H_{\omega}-sp+\frac{s^2}{2}\,,\textrm{~}
  e^{\ii sp}H_{\omega }e^{-\ii sp}
  = H_{\omega}+s\omega^2x+\frac{s^2\omega^2}{2}\,.
\end{equation}
By the well known spectral properties of $H_\omega$,
$\exp(-\ii{}TH_\omega)$ equals $(-1)^N\exp(-\ii\Delta{}H_\omega)$, and
so one finally arrives at the expression
\begin{displaymath}
  (-1)^N\exp\!\left(-\ii\gamma+\ii\,\frac{\alpha^2+\beta^2}{2\omega^2T}
    \left(1-\frac{\Delta}{T}\right)\right)
  \exp\!\left(-\ii\Delta H_\omega
    +\ii\,\frac{\alpha\Delta}{\omega T}\,p
    +\ii\,\frac{\beta\Delta}{T}\,x\right).
\end{displaymath}
Equating this expression to the RHS of (\ref{eq:UT_eq_exp}) one has to
set
\begin{displaymath}
  \alpha = \frac{T}{\Delta}\,\mu(t,0),\textrm{~}
  \beta = \frac{T}{\Delta}\,\nu(t,0),\textrm{~}
  -\gamma+\frac{\alpha^2+\beta^2}{2\omega^2T}
  \left(1-\frac{\Delta}{T}\right) = \sigma(T,0).
\end{displaymath}
Thus our choice of $H_F$ reads
\begin{equation}
  \label{eq:HF}
  H_F = H_{\omega}-\frac{\mu(T,0)}{\omega\Delta}\,p
  -\frac{\nu(T,0)}{\Delta}\,x-\frac{\sigma(T,0)}{T}
  +\pi N\,\frac{\mu(T,0)^2+\nu(T,0)^2}{\omega^3\Delta^2T}\,.
\end{equation}

As a next step one has to compute the $T$--periodic family of unitary
operators $U_F(t)=U(t,0)\exp(\ii{}tH_F)$. With the aid of
Lemma~\ref{thm:expH_eq_expp_expx_expH} one can express
\begin{eqnarray*}
  \exp\!\left(-\ii tH_F\right)
  &=& \exp\!\left(-\ii\phi(t)+\ii\,\frac{\sigma(T,0)t}{T}
    -\ii\pi N\,\frac{\big(\mu(T,0)^2
      +\nu(T,0)^2\big)t}{\omega^3\Delta^2T}\right)\\
  && \times\,\exp\!\left(\ii\,\frac{\xi(t)}{\omega}\,p\right)
   \exp\!\left(\ii\eta(t)x\right)\exp\!\left(-\ii tH_{\omega}\right)
\end{eqnarray*}
where
\begin{eqnarray*}
  \xi(t) &=& \frac{2}{\omega\Delta}\,
  \sin\!\left(\frac{\omega t}{2}\right)
  \left(\cos\!\left(\frac{\omega t}{2}\right)\mu(T,0)
    -\sin\!\left(\frac{\omega t}{2}\right)\nu(T,0)\right), \\
  \eta(t) &=& \frac{2}{\omega\Delta}\,
  \sin\!\left(\frac{\omega t}{2}\right)
  \left(\sin\!\left(\frac{\omega t}{2}\right)\mu(T,0)
    +\cos\!\left(\frac{\omega t}{2}\right)\nu(T,0)\right),
\end{eqnarray*}
and
\begin{eqnarray*}
  \phi(t) &=& -\frac{1}{4\omega^3\Delta^2}\,
  \big((2\omega t-4 \sin(\omega t)+\sin(2\omega t))\mu(T,0)^2 \\
  && +\,(2-4\cos(\omega t)+2\cos(2\omega t))\mu(T,0)\nu(T,0)
  +\,(2\omega t-\sin(2\omega t))\nu(T,0)^2\big).
\end{eqnarray*}
Using relations (\ref{eq:mu_nu}) for $\mu(T,0)$ and $\nu(T,0)$ this
can be rewritten as
\begin{eqnarray}
  \label{eq:xieta_eq_int}
  \xi(t) &=& \frac{\sin\!\left(\frac{\omega t}{2}\right)}{
    \sin\!\left(\frac{\omega T}{2}\right)}
  \int_0^T \sin\!\left(\omega\left(\frac{t+T}{2}-u\right)\right)
   f(u)\,\dd u, \nonumber\\
  \eta(t) &=& -\frac{\sin\!\left(\frac{\omega t}{2}\right)}{
    \sin\!\left(\frac{\omega T}{2}\right)}
  \int_0^T \cos\!\left(\omega\left(\frac{t+T}{2}-u\right)\right)
   f(u)\,\dd u,
\end{eqnarray}
and it also holds true that (compare to (\ref{eq:phi_eq_xieta}))
\begin{displaymath}
  \phi(t) = \frac{1}{2\omega}\,\xi(t)\eta(t)
  -\frac{\omega t-\sin(\omega t)}{8\omega
    \sin\!\left(\frac{\omega t}{2}\right)^2}
  \left(\xi(t)^2+\eta(t)^2\right).
\end{displaymath}

Expressing the propagator $U(t,0)$ according to formula
(\ref{eq:propag_EnssVeselic}) due to Enss and Veseli\`c one finally
arrives at the sought equality
\begin{equation}
  \label{eq:UF_ho}
  U_F(t) = e^{\ii\Phi(t)}e^{\ii F_2(t)x}e^{\ii(F_1(t)/\omega)p}
\end{equation}
where
\begin{displaymath}
  F_1(t) = \varphi_2(t,0)-\xi(t),\textrm{~}
  F_2(t) = -\varphi_1(t,0)-\eta(t),
\end{displaymath}
and
\begin{displaymath}
  \Phi(t) = -\psi(t,0)+\phi(t)-\frac{\sigma(T,0)t}{T}
  +\pi N\,\frac{\big(\mu(T,0)^2+\nu(T,0)^2\big)t}{\omega^3\Delta^2T}
  -\frac{\varphi_2(t,0)\eta(t)}{\omega}
\end{displaymath}
($\psi(t,s)$ is given in (\ref{eq:psi_def})). After some elementary
manipulations this can be rewritten as
\begin{eqnarray*}
  F_1(t) &=& \frac{1}{2\sin\!\left(\frac{\omega T}{2}\right)}
  \int_0^T \cos\!\left(\omega\left(u-\frac{T}{2}\right)\right)
  \big(f(t-u)-f(u)\big)\,\dd u, \\
  F_2(t) &=& \frac{1}{2\sin\!\left(\frac{\omega T}{2}\right)}
  \int_0^T \sin\!\left(\omega\left(u-\frac{T}{2}\right)\right)
  \big(f(t-u)+f(u)\big)\,\dd u,
\end{eqnarray*}
and
\begin{eqnarray*}
  \Phi(t) &=& -\frac{1}{2} \int_0^t
  \left(\varphi_1(v,0)^2-\varphi_2(v,0)^2\right)\,\dd v
  +\frac{t}{2 T}\,\int_0^T
  \left(\varphi_1(v,0)^2-\varphi _2(v,0)^2\right)\,\dd v \\
  && +\,\frac{1}{2\omega}\,\xi(t)\,\eta(t)
  -\frac{t}{2\omega T}\,\varphi_1(T,0)\varphi_2(T,0)
  -\frac{\varphi_2(t,0)\eta(t)}{\omega} \\
  && -\,\frac{\omega t-\sin(\omega t)}
  {8\omega\sin\!\left(\frac{\omega t}{2}\right)^{\!2}}
  \left(\xi(t)^2+\eta(t)^2\right)
  +\frac{(\omega T-\sin(\omega T))t}{8\omega
    \sin\!\left(\frac{\omega T}{2}\right)^{\!2}T}
  \left(\varphi_1(T,0)^2+\varphi _2(T,0)^2\right).
\end{eqnarray*}
In the last equality one has to substitute for $\varphi_1(t,0)$ and
$\varphi_2(t,0)$ from (\ref{eq:varphi12}), and for $\xi(t)$ and
$\eta(t)$ from (\ref{eq:xieta_eq_int}).

It is of importance to observe that the functions $F_1(t)$, $F_2(t)$
and $\Phi(t)$ entering formula (\ref{eq:UF_ho}) are continuously
differentiable. In addition, they are necessarily $T$--periodic.
Furthermore, the operators $x$ and $p$ are infinitesimally small with
respect to $H_\omega$. This is a well known fact which is also briefly
recalled in the beginning of the Appendix.  From equality
(\ref{eq:HF}) one can see that $\Dom{}H_F=\Dom{}H_\omega$. Moreover,
from the commutation relations (\ref{eq:expsHexps}) it follows that
the unitary groups $\{\exp(\ii{}sx);\,s\in\R\}$ and
$\{\exp(\ii{}sp);\,s\in\R\}$ preserve the domain $\Dom{}H_\omega$.
Hence one can differentiate $U_F(t)$ given in (\ref{eq:UF_ho}) on any
vector $\psi\in\Dom{}H_F$. Computing $S_F(t)$ according to
(\ref{eq:S_eq_diff_UF}) one finds that
\begin{displaymath}
  S_F(t) = -\frac{F_1'(t)}{\omega}\,p-F_2'(t)x
  +\frac{F_1(t)F_2'(t)}{\omega}-\Phi'(t).
\end{displaymath}
Consequently, $S_F(t)$ is infinitesimally small with respect to $H_F$
for any $t$. Thus all assumptions of
Proposition~\ref{prop:energy_bounded} are fulfilled and one concludes
that $\|H(t)U(t,0)\psi\|$ is bounded in time for any
$\psi\in\Dom(H(0))=\Dom(H_\omega)$.

From the explicit form of $H(t)$ and from the infinitesimal smallness
of $x$ with respect to $H_\omega$ it follows that the quantity
$\|H_\omega{}U(t,0)\psi\|$ is bounded in time as well. Let us recall
once more the consequences of this observation. Firstly, as stressed
in
\cite[Proposition~4]{deoliveira95:_some_rems_stabil_nonstationary_qs},
since $F(H_\omega<R)$ is a finite rank projector for any $R>0$ it is
true that all trajectories $\{U(t,0)\psi;\,t\in\R\}$ are precompact.
Secondly, in virtue of Theorem~2.3 in
\cite{enss_veselic83:_bound_propag_states_time_dep_hamil}, the
monodromy operator $U(T,0)$ has a pure point spectrum.

Finally, let us shortly discuss the resonant case $T=(2\pi/\omega)N$,
$N\in\N$. Using again formula (\ref{eq:propag_EnssVeselic}) we have
\begin{equation}
  \label{eq:UT_resonant}
  U(T,0) = (-1)^Ne^{-\ii\psi(T,0)}
  \exp\!\left(-\ii\varphi_1(T,0)x\right)
  \exp\!\left(\ii\,\frac{\varphi_2(T,0)}{\omega}\,p\right).
\end{equation}
Notice that the unitary operator $e^{\ii\alpha{}x}e^{\ii\beta{}p}$,
with $\alpha,\beta\in\R$, is either the identity if $\alpha=\beta=0$
or it has a purely absolutely continuous spectrum. For example, if
$\beta\neq0$ then we have the commutation relation
\begin{displaymath}
  e^{\ii\alpha{}x}e^{\ii\beta{}p}
  = \exp\!\left(-\ii\,\frac{\alpha}{2\beta}\,x^2\right)
  \exp\!\left(\ii\beta p-\frac{\ii}{2}\,\alpha\beta\right)
  \exp\!\left(\ii\,\frac{\alpha}{2\beta}\,x^2\right).
\end{displaymath}
Hence the spectrum of $e^{\ii\alpha{}x}e^{\ii\beta{}p}$ coincides with
that of $e^{-\ii\alpha\beta/2}e^{\ii\beta{}p}$. In the case
$\alpha\neq0$ one can argue in a similar way. Thus when applying this
observation to (\ref{eq:UT_resonant}) we have to distinguish the case
$\varphi_1(T,0)=\varphi_2(T,0)=0$. Recalling defining relations
(\ref{eq:varphi12}) we denote by
\begin{displaymath}
  f_k = \frac{1}{T}\int_0^T
  \exp\!\left(-\ii\,\frac{2\pi}{T}\,kt\right)f(t)\,\dd t,\textrm{~~}
  k\in\Z,
\end{displaymath}
the Fourier coefficients of $f(t)$. We conclude that if $f_{-N}=f_N=0$
then the monodromy operator $U(T,0)$, with $T=2\pi{}N/\omega$, is a
multiple of the identity. If $|f_{-N}|+|f_N|>0$ then $U(T,0)$ has a
purely absolutely continuous spectrum.  This in turn implies that, in
the latter case, the quantity $\|H(t)U(t,0)\psi\|$ cannot happen to be
bounded in time for all $\psi\in\Dom{}H_\omega$.

\section{An application of the quantum KAM method}%

The quantum KAM method was originally proposed by Bellissard
\cite{bellissard85:_stability_instability_qm} and it has been later
reconsidered and in some respects improved several times, see for
example \cite{combescure87:_quant_stability_periodic_perturb_harmosc,
  bleher_jauslin_lebowitz92:_floquet_spectrum_quasi_period,
  duclos_ps96:_floquet_hamiltonian_pure_point,
  bambusi_graffi01:_time_quasi_period_schroedinger_kam,
  duclos_lev_ps_vittot02:_weakly_regular_floquet_hamil}. When
discussing an application of the quantum KAM method to our problem we
shall stick to the presentation given in
\cite{duclos_lev_ps_vittot02:_weakly_regular_floquet_hamil} but the
notation will be partially modified. A particularity of the method is
that the frequency $\omega=2\pi/T$ should be considered as a
parameter. Usually the method is used to show that for a large subset
of so called non-resonant frequencies the spectrum of the Floquet
Hamiltonian is pure point. Here we would like to point out, following
some ideas from
\cite{asch_duclos_exner98:_stability_growing_gaps_wannier}, that the
method provides a more detailed information which can be used to
reveal the structure of the propagator.

Let us first recall the main theorem from
\cite{duclos_lev_ps_vittot02:_weakly_regular_floquet_hamil}. Let
$H_{0}$ be a self-adjoint operator in $\sH$ with a discrete spectrum,
$\Spec(H_{0})=\{ h_{m}\}_{m=1}^{\infty}$, and such that the
multiplicities $M_{m}=\dim\Ker(H_{0}-h_{m})$ are finite. Suppose also
that
\[
\Delta_{0}=\inf_{m\neq n}|h_{m}-h_{n}|>0.
\]
Furthermore, let $V(t)$ be a $2\pi$-periodic uniformly bounded
operator-valued function defined on $\R$ and with values in
$\sB(\sH)$. Set
\[
V_{knm}=\frac{1}{2\pi}\,\int_{0}^{2\pi}e^{-\ii kt}Q_{n}V(t)Q_{m}\dd t
\]
where $Q_{n}$ is the orthogonal projector onto $\Ker(H_{0}-h_{n})$.
As already mentioned, the frequency $\omega=2\pi/T$, $T>0$, is
regarded as a parameter. Set $\sK=L^{2}([\,0,T\,],\sH,\dd t)$ and let
$V\in\sB(\sK)$ be the operator acting via multiplication by
$V(\omega{}t)$, $(Vf)(t)=V(\omega t)f(t)$.  Let $K_{0}$ be the closure
of $-\ii\partial_{t}\otimes1+1\otimes H_{0}$.

\begin{thm}\label{prop:RevMathPhys}
  Fix $J>0$ and set $\Omega_{0}=[\,\frac{8}{9}J,\frac{9}{8}J\,]$.
  Assume that there exists $\sigma>0$ such that
  \begin{equation}
    \label{eq:gap_cond_sigma}
    \sum_{\substack{m,n\in\N\cr h_{m}-h_{n}>J/2}}
    \frac{M_{m}M_{n}}{(h_{m}-h_{n})^{\sigma}}<\infty\,.
  \end{equation}
  Then for every $r>\sigma+\frac{1}{2}$ there exist positive constants
  (depending on $\sigma$, $r$, $\Delta_{0}$ and $J$ but independent of
  $V$), $\epsilon_{\star}$ and $\delta_{\star}$, with the property:

  \noindent if
  \begin{equation}
    \label{eq:eps_V}
    \epsilon_{V}:=\sup_{n\in\N}\sum_{m\in\N}\sum_{k\in\Z}(1+|k|)^{r}
    \Vert V_{knm}\Vert<\epsilon_{\star}
  \end{equation}
  then there exists a measurable subset
  $\Omega_{\infty}\subset\Omega_{0}$ such that
  \[
  |\Omega_{\infty}|\geq|\Omega_{0}|-\delta_{\star}\,\epsilon_{V}
  \]
  (here $|\Omega_{\ast}|$ stands for the Lebesgue measure of
  $\Omega_{\ast}$) and the operator $K_{0}+V$ has a pure point
  spectrum for all $\omega\in\Omega_{\infty}$.
\end{thm}

The proof of Theorem~\ref{prop:RevMathPhys} is somewhat lengthy and
tedious because one has to eliminate the resonant frequencies.  The
basic idea is, however, rather simple and is based on an iterative
procedure as described in the following proposition. It is formulated
even on the level of Banach spaces but afterwards we shall again work
with Hilbert spaces. Let $\Phi(x)$ be the analytic function defined by
\[
\Phi(x) = \frac{1}{x}\left(e^{x}-\frac{e^{x}-1}{x}\right).
\]

\begin{prop}
  \label{prop:iter_diag_proc}
  Assume that $\sK$ is a Banach space, $K_{0}$ is a closed operator in
  $\sK$, $V\in\sB(\sK)$ and $D\in\sB(\sB(\sK))$. Assume further that
  $V=\lim V_{s}$ in $\sB(\sK)$ where $\{V_{s}\}_{s=0}^{\infty}$ is a
  sequence of bounded operators in $\sK$. If there exist sequences
  $\{A_{s}\}_{s=0}^{\infty}$ and $\{ G_{s}\}_{s=0}^{\infty}$,
  $A_{s},G_{s}\in\sB(\sK)$, fulfilling the following recurrence
  relations for all $s\in\Z_{+}$:
  \begin{eqnarray}
    G_{0} & = & V_{0},\nonumber \\
    G_{s+1} & = & G_{s}+\exp(\ad_{A_{s}})\ldots
    \exp(\ad_{A_{0}})(V_{s+1}-V_{s})\label{eq:recruleGs}\\
    &  & +\,\ad_{A_{s}}\Phi(\ad_{A_{s}})(1-D)(G_{s}-G_{s-1}),\nonumber
  \end{eqnarray}
  $A_{s}\Dom(K_{0})\subset\Dom(K_{0})$,
  \begin{eqnarray}
    [A_{0},K_{0}+D(G_{0})] & = & -(1-D)(G_{0}),\nonumber \\
    \mbox{}[A_{s+1},K_{0}+D(G_{s+1})] & = & -(1-D)(G_{s+1}-G_{s}),
    \label{eq:recruleAs}
  \end{eqnarray}
  and such that $\sum_{s=0}^{\infty}\Vert A_{s}\Vert<\infty$ and the
  limit $\lim G_{s}=G_{\infty}$ exists in $\sB(\sK)$ then there exists
  $W\in\sB(\sK)$ such that $W^{-1}\in\sB(\sK)$ and
  \begin{equation}
    \label{eq:W_diag_K0plusV}
    W(K_{0}+V)W^{-1}=K_{0}+D(G_{\infty}).
  \end{equation}
\end{prop}

Here, as usual, $\ad$ means the adjoint action, $\ad_{A}X=[A,X]$ and
$\exp(\ad_{A})X=e^{A}Xe^{-A}$. For $s=0$ in (\ref{eq:recruleGs}) we
set $G_{-1}=0$. The proof of Proposition~\ref{prop:iter_diag_proc} is
immediate. If we set
\[
W_{s}=e^{A_{s-1}}\ldots e^{A_{0}},\textrm{~}W_{s}^{\textrm{~}-1}
= e^{-A_{0}}\ldots e^{-A_{s-1}},
\]
then the recurrence relations (\ref{eq:recruleGs}),
(\ref{eq:recruleAs}) exactly mean that
\[
\forall s\in\Z_{+},\textrm{~}W_{s}(K_{0}+V_{s})W_{s}^{\textrm{~}-1}
= K_{0}+D(G_{s})+(1-D)(G_{s}-G_{s-1}).
\]
Now it suffices to send $s$ to infinity.

In the applications of Proposition~\ref{prop:iter_diag_proc}, and this
is also the case for Theorem~\ref{prop:RevMathPhys}, $\sK$ is a
separable Hilbert space, $K_{0}=K_{0}^{\textrm{~}\ast}$, $V=V^{\ast}$,
the spectrum of $K_{0}$ is pure point and $D(X)$ is the diagonal part
of a bounded operator $X$ with respect to the spectral decomposition
of $K_{0}$. Then $G_{\infty}^{\textrm{~}\ast}=G_{\infty}$,
$D(G_{\infty})^{\ast}=D(G_{\infty})$ and $W^{\ast}=W^{-1}$. The
operator $K_{0}+D(G_{\infty})$ has obviously a pure point spectrum and
relation (\ref{eq:W_diag_K0plusV}) implies that the same is true for
$K_{0}+V$.

Let us note that technically the basic problem of the entire method is
the commutator equation (\ref{eq:recruleAs}) whose solution is
complicated by the fact that, generically, the eigenvalues of $K_{0}$
are dense in $\R$. This leads to the famous problem of small
denominators in this context.

There is another feature concerning the application of the recursive
procedure (\ref{eq:recruleGs}) and (\ref{eq:recruleAs}) in the proof
of Theorem~\ref{prop:RevMathPhys}. Let $M\in\sB(\sK)$ be the
multiplication operator defined by the relation
\begin{equation}
  \label{eq:M_def}
  \forall f\in\sK,\textrm{~}(Mf)(t)=e^{\ii\omega t}f(t).
\end{equation}
Since $V\in\sB(\sK)$ is a multiplication operator it commutes with
$M$. Also the sequence $\{ V_{s}\}$ is chosen in such a way that $M$
commutes with all $V_{s}$. Furthermore, the eigenvalues of $K_{0}$ are
$k\omega+h_{m}$, $k\in\Z$ and $m\in\N$, and so they are linear in $k$.
Using these facts it is readily seen from the recursive relations that
$M$ commutes with both $A_{s}$ and $G_{s}$ for all $s$. Then
necessarily $M$ commutes with $G_{\infty}$ and $W$ as well. This
implies that there exists a bounded Hermitian operator $G$ on $\sH$
such that $G\Dom(H_{0})\subset\Dom(H_{0})$,
\begin{displaymath}
  [H_{0},G] = 0\textrm{~and~}\big(D(G_{\infty})f\big)(t) = Gf(t),\quad
  \forall f\in\sK,\textrm{~}a.a.\textrm{~}t\in\R,
\end{displaymath}
and there exists a $T$-periodic operator-valued function $t\mapsto
W(t)$ with values in unitary operators on $\sH$ such that equality
(\ref{eq:W_diag_K0plusV}) is satisfied with
\[
(Wf)(t)=W(t)f(t),\quad\forall f\in\sK,\textrm{~}a.a.\textrm{~}t\in\R.
\]

Moreover, an information about the regularity of $W$ is also available.
More precisely, one knows that
\begin{equation}
  \label{eq:regW}
  \sup_{n\in\N}\sum_{m\in\N}\sum_{k\in\Z}\Vert W_{knm}\Vert<\infty
\end{equation}
where again
\begin{displaymath}
  W_{knm}=\frac{1}{T}\,\int_{0}^{T}
  e^{-\ii k\omega t}Q_{n}W(t)Q_{m}\dd{}t\,.
\end{displaymath}
Particularly, the operator-valued function $W(t)$ is continuous even
in the operator norm. Equality (\ref{eq:W_diag_K0plusV}) can be
rewritten in terms of propagators. It exactly means that
\begin{equation}
  \label{eq:Uts_from_KAM}
  \forall t,s\in\R,\textrm{~}U(t,s)
  =W(t)^{\ast}e^{-\ii(t-s)(H_{0}+G)}W(s).
\end{equation}

By a closer look at the proof of Theorem~\ref{prop:RevMathPhys} one
finds that the result can be partially improved. In the course of the
proof one constructs a directed sequence of Banach spaces
$\{\gX_{s}\}$,
\[
\gX_{s} \subset L^{\infty}\!\left(\Omega_{s}\times\Z\times\N\times\N,\,
  \sum_{n\in\N}\sideset{}{^{\oplus}}
  \sum_{m\in\N}\sB(\HH_{m},\HH_{n})\right),
\]
with the norms
\begin{equation}
  \label{eq:norm_Xs}
  \Vert X\Vert_{s}=\sup_{\substack{\omega,\omega'\in\Omega_{s}\cr
      \omega\neq\omega'}}\,
  \sup_{n\in\N}\,\sum_{m\in\N}\sum_{k\in\Z}
  \left(\Vert X_{knm}(\omega)\Vert+\varphi_{s}\,
    \Vert\widetilde{\partial}X_{knm}(\omega,\omega')
    \Vert\right)e^{|k|/\EE_{s}}
\end{equation}
where $X=\{ X_{knm}(\omega)\}\in\gX_{s}$, i.e.,
$X_{knm}(\omega)\in\sB(\HH_{m},\HH_{n})$ for all $\omega\in\Omega_{s}$
and $(k,n,m)\in\Z\times\N\times\N$. Here
$\HH_{m}:=\Ker(H_{0}-h_{m})=\Ran Q_{m}$, $\{\Omega_{s}\}$ is a
decreasing sequence of subsets of the interval $\Omega_{0}$,
$\{\varphi_{s}\}$ and $\{\EE_{s}\}$ are respectively decreasing and
strictly increasing sequences of positive numbers such that
$\lim\varphi_{s}=0$, $1\leq\EE_{s}$ and $\lim\EE_{s}=+\infty$. The
symbol $\widetilde{\partial}$ designates the discrete derivative in
$\omega$,
\[
\widetilde{\partial}X(\omega,\omega')
=\frac{X(\omega)-X(\omega')}{\omega-\omega'}\,.
\]
For $\omega\in\Omega_{\infty}=\bigcap\Omega_{s}$ fixed one applies the
limit procedure $s\to\infty$ and arrives at equality
(\ref{eq:W_diag_K0plusV}) with the objects $G_{\infty}$ and $W$
belonging to the Banach space
\[
\gX_{\infty}\subset L^{\infty}\!\left(\Z\times\N\times\N,\,
  \sum_{n\in\N}\sideset{}{^{\oplus}}
  \sum_{m\in\N}\sB(\HH_{m},\HH_{n})\right)
\]
where the norm is defined by
\[
\Vert X\Vert_{\infty}=\sup_{n\in\N}\,
\sum_{m\in\N}\sum_{k\in\Z}\Vert X_{knm}(\omega)\Vert\,.
\]
This is also how one obtains the information about the regularity of
$W$ expressed in (\ref{eq:regW}).

The announced improvement consists in modifying the norms
(\ref{eq:norm_Xs}) by an additional weight $(1+|k|)^{\nu}$ where $\nu$
should be chosen in the range
\begin{equation}
  \label{eq:nu_leq_r-sgm-half}
  0\leq\nu<r-\sigma-\frac{1}{2}\,.
\end{equation}
Recall that $r$ determines the regularity of $V$ in (\ref{eq:eps_V}),
$\sigma$ comes from the {}``gap condition'' (\ref{eq:gap_cond_sigma})
and one requires that $r>\sigma+\frac{1}{2}$. The modified norm reads
\begin{equation}
  \Vert X\Vert_{s}=\sup_{\substack{\omega,\omega'\in\Omega_{s}\cr
      \omega\neq\omega'}}\,\sup_{n\in\N}\,\sum_{m\in\N}
  \sum_{k\in\Z}(1+|k|)^{\nu}
  \left(\Vert X_{knm}(\omega)\Vert+\varphi_{s}\,
    \Vert\widetilde{\partial}X_{knm}(\omega,\omega')
    \Vert\right)e^{|k|/\EE_{s}}\,,
  \label{eq:norm_modif}
\end{equation}
and the limit procedure results in a norm in $\gX_{\infty}$,
\begin{equation}
  \Vert X\Vert_{\infty}=\sup_{n\in\N}\,\sum_{m\in\N}
  \sum_{k\in\Z}(1+|k|)^{\nu}\Vert X_{knm}(\omega)\Vert\,.
  \label{eq:norm_infty_modified}
\end{equation}
Let us note that restriction (\ref{eq:nu_leq_r-sgm-half}) comes from
the lower estimate of Lebesgue measure of the set $\Omega_{\infty}$
(see relation (77) in
\cite{duclos_lev_ps_vittot02:_weakly_regular_floquet_hamil} and the
derivation preceding it where one has to replace $r$ by $r-\nu$ if
using the modified norm (\ref{eq:norm_modif})). After this
modification, Theorem~\ref{prop:RevMathPhys} is valid exactly in the
same formulation as before, its proof requires no additional changes,
only the constants $\epsilon_{\star}$ and $\delta_{\star}$ should be
modified correspondingly.

The interest of the modification is that we get a better information
about the regularity of $W$. Namely, for $\omega\in\Omega_{\infty}$
(the set of non-resonant frequencies) $W$ is regular in the sense that
$\Vert W\Vert_{\infty}<\infty$ with the norm given by
(\ref{eq:norm_infty_modified}). In particular, if
$r>\sigma+\frac{3}{2}$ then one can choose $\nu\geq1$.  In that case
the property $\Vert W\Vert_{\infty}<\infty$ implies that $W(t)$
belongs to the class $C^{1}$ in the operator norm and
$\sup_{t}\Vert\partial_{t}W(t)\Vert<\infty$.

This discussion shows that Theorem~\ref{prop:RevMathPhys} can be
reformulated in the following way.

\begin{thm}\label{prop:RMP_reform}
  Under the same assumptions as in Theorem~\ref{prop:RevMathPhys}
  suppose that $r>\sigma+\frac{3}{2}$. Then there exist positive
  constants (independent of $V$), $\epsilon_{\star}$ and
  $\delta_{\star}$, with the property:

  if $\epsilon_{V}<\epsilon_{\star}$ then there exists a measurable
  subset $\Omega_{\infty}\subset\Omega_{0}$ such that
  $|\Omega_{\infty}|\geq\linebreak
  |\Omega_{0}|-\delta_{\star}\,\epsilon_{V}$,
  and for every $\omega\in\Omega_{\infty}$ there exist a bounded
  Hermitian operator $G$ commuting with $H_{0}$ and a $T$-periodic
  function $W(t)$ with values in unitary operators and belonging to
  the class $C^{1}$ in the operator norm such that the propagator
  obeys equality (\ref{eq:Uts_from_KAM}).
\end{thm}

From relation (\ref{eq:Uts_from_KAM}) it follows that the propagator
admits a Floquet decomposition (\ref{eq:Floquet-decomp}) with
\[
H_{F} = W(0)^{\ast}(H_{0}+G)W(0),\textrm{~}U_{F}(t)=W(t)^{\ast}W(0).
\]
Moreover, formula (\ref{eq:S_eq_diff_UF}) implies that
\[
S_F(t) = -\ii\, W(0)^{\ast}
\left(\partial_{t}W(t)\right)W(t)^{\ast}W(0).
\]
In particular, if $W(t)$ is known to be $C^{1}$ in the operator norm
then the Floquet decomposition is continuously differentiable in the
strong sense and, consequently, the assumptions both of
Proposition~\ref{prop:energy_bounded} and
Proposition~\ref{prop:trans_probability} are satisfied. These
arguments prove the following theorem.

\begin{thm}
  Under the same assumptions as in Theorem~\ref{prop:RevMathPhys}
  suppose that $r>\sigma+\frac{3}{2}$. Then there exist positive
  constants (independent of $V$), $\epsilon_{\star}$ and
  $\delta_{\star}$, with the property:

  if $\epsilon_{V}<\epsilon_{\star}$ (with $\epsilon_{V}$ defined in
  (\ref{eq:eps_V})) then there exists a measurable subset
  $\Omega_{\infty}\subset\Omega_{0}$ such that
  $|\Omega_{\infty}|\geq|\Omega_{0}|-\delta_{\star}\,\epsilon_{V}$ and
  for every $\omega\in\Omega_{\infty}$ the energy of the quantum
  system described by the time-dependent Hamiltonian
  $H_{0}+V(\omega{}t)$ is bounded uniformly in time. More precisely,
  \[
  \forall\psi\in\Dom(H(0)),\textrm{~}
  \sup_{t\in\R}\Vert(H_{0}+V(\omega t))U(t,0)\psi\Vert < \infty.
  \]
  Moreover, there exists a constant $c\geq0$ such that for any couple
  of intervals $\Delta_{1},\Delta_{2}\subset\R$ fulfilling
  $\dist(\Delta_{1},\Delta_{2})>0$ it holds true that
  \[
  \sup_{s,t\in\R}\Vert P(t,\Delta_{1})U(t,s)P(s,\Delta_{2})\Vert
  \leq\frac{c}{\dist(\Delta_{1},\Delta_{2})}
  \]
  where $P(t,\cdot)$ is the spectral measure of $H(t)$. In particular,
  if $E_{n}(t)$ and $E_{m}(s)$ are two distinct eigenvalues of $H(t)$
  and $H(s)$, respectively, and $P_{n}(t)$ and $P_{m}(s)$ are the
  corresponding orthogonal projectors then
  \[
  \Vert
  P_{n}(t)U(t,s)P_{m}(s)\Vert\leq\frac{c}{|E_{n}(t)-E_{m}(s)|}\,.
  \]
\end{thm}

\section*{Acknowledgments}
P.~\v{S}. wishes to acknowledge gratefully the support from the grant
No. 201/05/0857 of Grant Agency of the Czech Republic.

\newpage
\setcounter{section}{1}
\renewcommand{\thesection}{\Alph{section}}
\setcounter{equation}{0}
\renewcommand{\theequation}{\Alph{section}.\arabic{equation}}

\section*{Appendix. The propagator for the time-dependent harmonic
  oscillator}

Let $H_\omega=(1/2)(p^{\,2}+\omega^2x^2)$, with $p=-\ii\partial_x$, be
the Hamiltonian of the harmonic oscillator in $\sH=L^2(\R,\dd{}x)$.
This is a known fact that the operators $x$ and $p$ are relatively
bounded with respect to $H_\omega$ with the relative bound zero. One
can see it also directly from the following inequality which is easy
to verify on the Schwarz space $\sS$:
\begin{displaymath}
  2\omega^2 x^2 \leq \varepsilon^{-2}+\varepsilon^2\omega^4x^4
  +\varepsilon^2\left(p^4+2\omega^2p\,x^2p\right)
  = \varepsilon^{-2}+2\varepsilon^2\omega^2
  +4\varepsilon^2H_{\omega}^{\,2}.
\end{displaymath}
Hence it is true that
\begin{displaymath}
  \|x\psi\|^2
  \leq \left(\varepsilon^2+\frac{1}{2\omega^2\varepsilon^2}\right)
  \|\psi\|^2+\frac{2\varepsilon^2}{\omega^2}
  \left\|H_{\omega}\psi\right\|^2
\end{displaymath}
for all $\psi\in\sS$ and consequently for all $\psi\in\Dom{}H_\omega$.
A bound for the operator $p$ can be derived analogously. It follows
that for any $\alpha,\beta\in\R$, the domain
$\Dom(H_\omega+\alpha{}x+\beta{}p)$ coincides with $\Dom{}H_\omega$.

For $\varphi,\psi\in\Dom(H_\omega)$ set
\begin{displaymath}
  \bx(t)
  = \langle e^{-\ii tH_\omega}\varphi,xe^{-\ii tH_\omega}\psi\rangle,
  \textrm{~}\bp(t)
  = \langle e^{-\ii tH_\omega}\varphi,pe^{-\ii tH_\omega}\psi\rangle.
\end{displaymath}
As a standard exercise one derives, by differentiating and using the
canonical commutation relation, that the quantities $\bx(t)$ and
$\bp(t)$ obey the classical evolution equations, i.e.,
$\dot{\bx}(t)=\bp(t)$, $\dot{\bp}(t)=-\omega^2\,\bx(t)$. It follows
that for all $\psi\in\Dom(H_\omega)$,
\begin{eqnarray}
  \label{eq:xt_pt_classic}
  e^{\ii tH_\omega}xe^{-\ii tH_\omega}\psi
  &=& \cos(\omega t)x\psi+\frac{1}{\omega}\,\sin(\omega t)\,p\psi,
  \nonumber\\
  e^{\ii tH_\omega}pe^{-\ii tH_\omega}\psi
  &=& -\omega\sin(\omega t)x\psi+\cos(\omega t)\,p\psi.
\end{eqnarray}

\begin{alem}
  \label{thm:expH_eq_expp_expx_expH}
  For $\mu,\nu,t\in\R$ it holds
  \begin{equation}
    \label{eq:expH_eq_expp_expx_expH}
    \exp\!
    \left(-\ii tH_\omega+\ii\,\frac{\mu}{\omega}\,p+\ii\nu x\right)
    = e^{-\ii\phi}\exp\!\left(\ii\,\frac{\xi}{\omega}\,p\right)
    \exp(\ii\eta x)\exp\!\left(-\ii tH_\omega\right)
  \end{equation}
  where
  \begin{eqnarray}
    \label{eq:xi_eta}
    \xi &=& \frac{2\sin\!\left(\frac{\omega t}{2}\right)}{\omega t}
    \left(\cos\!\left(\frac{\omega t}{2}\right)\mu
      -\sin\!\left(\frac{\omega t}{2}\right)\nu\right),
    \nonumber\\
    \eta &=& \frac{2\sin\!\left(\frac{\omega t}{2}\right)}{\omega t}
    \left(\sin\!\left(\frac{\omega t}{2}\right)\mu
      +\cos\!\left(\frac{\omega t}{2}\right)\nu\right),
  \end{eqnarray}
  and
  \begin{eqnarray}
    \label{eq:phi}
    \phi &=& -\frac{1}{4\omega^3t^2}
    \left((2\omega t-4 \sin(\omega t)+\sin(2\omega t))\mu^2
    +(2-4\cos(\omega t)+2\cos(2\omega t))\mu\nu\right. \nonumber\\
    && \left.+\,(2\omega t-\sin(2\omega t))\nu^2\right).
  \end{eqnarray}
\end{alem}

\begin{proof}
  Set
  \begin{displaymath}
    W_1 = \exp\!\left(\ii\,\frac{\xi}{\omega}\,p\right)\exp(\ii\eta x)
    \exp\!\left(-\ii tH_\omega\right),\textrm{~}
    W_2 = \exp\!
    \left(-\ii tH_\omega+\ii\frac{\mu}{\omega}\,p+\ii\nu x\right).
  \end{displaymath}
  Clearly, both $W_1$ and $W_2$ leave the domain of $H_\omega$
  invariant. Using (\ref{eq:xt_pt_classic}) one finds that for all
  $\psi\in\Dom(H_\omega)$ it holds
  \begin{eqnarray}
    \label{eq:W1xpW1}
    W_1^{\,-1}xW_1\psi &=& \cos(\omega t)x\psi
    +\frac{1}{\omega}\,\sin(\omega t)\,p\psi-\frac{\xi}{\omega}\,\psi,
    \nonumber\\
    W_1^{\,-1}pW_1\psi &=& -\omega\sin(\omega t)x\psi
    +\cos(\omega t)\,p\psi+\eta\psi.
  \end{eqnarray}
  On the other hand,
  \begin{displaymath}
    -\ii tH_\omega+\ii\frac{\mu}{\omega}\,p+\ii\nu x
    = -\ii\,\frac{t}{2}
    \left(\widetilde{p}^{\,2}+\omega^2\widetilde{x}^2\right)
    +\ii\,\frac{\mu^2+\nu^2}{2\,\omega^2t}
  \end{displaymath}
  where
  \begin{displaymath}
    \widetilde{x} = x-\frac{\nu}{\omega^2t}\,,\textrm{~}
    \widetilde{p} = p-\frac{\mu}{\omega t}\,.
  \end{displaymath}
  Since $\widetilde{p}$ and $\widetilde{x}$ also obey the canonical
  commutation relation we have, analogously to
  (\ref{eq:xt_pt_classic}),
  \begin{eqnarray*}
    W_2^{\,-1}\widetilde{x}\,W_2\psi
    &=& \cos(\omega t)\widetilde{x}\,\psi
    +\frac{1}{\omega}\,\sin(\omega t)\,\widetilde{p}\,\psi, \\
    W_2^{\,-1}\widetilde{p}\,W_2\psi
    &=& -\omega\sin(\omega t)\widetilde{x}\,\psi
    +\cos(\omega t)\,\widetilde{p}\,\psi.
  \end{eqnarray*}
  for all $\psi\in\Dom(H_\omega)$. This can be rewritten as
  \begin{eqnarray}
    \label{eq:W2xpW2}
    W_2^{\,-1}xW_2\psi &=& \cos(\omega t)x\psi
    +\frac{1}{\omega}\,\sin(\omega t)p\psi
    +\frac{1}{\omega^2t}
    \left(\nu-\nu\cos(\omega t)-\mu\sin(\omega t)\right)\psi,
    \nonumber\\
    W_2^{\,-1}pW_2\psi
    &=& -\omega\sin(\omega t)x\psi
    +\cos(\omega t)p\psi+\frac{1}{\omega t}
    \left(\mu+\nu\sin(\omega t)-\mu\cos(\omega t)\right)\psi.
    \nonumber\\
  \end{eqnarray}
  Comparing (\ref{eq:W1xpW1}) to (\ref{eq:W2xpW2}) one finds that for
  all $\psi\in\Dom(H_\omega)$ it holds
  $W_1^{\,-1}xW_1\psi=W_2^{\,-1}xW_2\psi$ and
  $W_1^{\,-1}pW_1\psi=W_2^{\,-1}pW_2\psi$ provided
  \begin{displaymath}
    \xi = -\frac{\nu}{\omega t}\big(1-\cos(\omega t)\big)
    +\frac{\mu}{\omega t}\,\sin(\omega t),\textrm{~}
    \eta = \frac{\nu}{\omega t}\,\sin(\omega t)
    +\frac{\mu}{\omega t}\big(1-\cos(\omega t)\big)
  \end{displaymath}
  (which is nothing but (\ref{eq:xi_eta})). Hence $W=W_2W_1^{\,-1}$
  fulfills $W^{-1}xW\psi=x\psi$ and $W^{-1}pW\psi=p\psi$. Since
  $\Dom(H_\omega)$ is a core both for $x$ and $p$ this implies that
  $W^{-1}xW=x$ and $W^{-1}pW=p$. If follows that $W$ is a multiple of
  the unity, i.e., $W_2=e^{-\ii\phi}W_1$ for some $\phi\in\R$.

  It remains to determine $\phi$. To this end it suffices to take the
  mean value of the corresponding operators at the ground state of
  $H_\omega$ which is $\psi_0(x)=\exp(-\omega{}x^2/2)$ (unnormalized).
  Writing
  \begin{displaymath}
    W_2 = \exp\!\left(\ii\,\frac{\mu^2+\nu^2}{2\,\omega^2t}\right)
    \exp\!\left(-\ii\,\frac{\nu}{\omega^2t}\,p\right)
    \exp\!\left(\ii\,\frac{\mu}{\omega t}\,x\right)
    e^{-\ii t H_\omega}
    \exp\!\left(-\ii\,\frac{\mu}{\omega t}\,x\right)
    \exp\!\left(\ii\,\frac{\nu}{\omega^2t}\,p\right)
  \end{displaymath}
  the equality $W_2=e^{-\ii\phi}W_1$ becomes
  \begin{eqnarray*}
    && \exp\!\left(\ii\,\frac{\mu^2+\nu^2}{2\,\omega^2t}\right)
    e^{-\ii t H_\omega}
    \exp\!\left(-\ii\,\frac{\mu}{\omega t}\,x\right)
    \exp\!\left(\ii\,\frac{\nu}{\omega^2t}\,p\right) \\
    && =\, \exp\!\left(-\ii\phi+\ii\left(\frac{\nu}{\omega^2t}
        +\frac{\xi}{\omega}\right)\eta\right)
    \exp\!\left(\ii\left(\eta-\frac{\mu}{\omega t}\right)x\right)
    \exp\!\left(\ii\left(\frac{\xi}{\omega}
        +\frac{\nu}{\omega^2t}\right)p\right)
    e^{-\ii t H_\omega}.
  \end{eqnarray*}
  The mean value at $\psi_0$ of the both sides of the last equality
  then yields
  \begin{eqnarray*}
    && \exp\!\left(\ii\left(\phi+\frac{\mu^2+\nu^2}{2\,\omega^2t}
        -\left(\frac{\nu}{\omega^2t}
          +\frac{\xi}{\omega}\right)\eta\right)\right)
    \left\langle\psi_0(x),
      \exp\!\left(-\ii\,\frac{\mu}{\omega t}\,x\right)
      \psi_0\!\left(x+\frac{\nu}{\omega^2t}\right)\right\rangle \\
    && =\, \left\langle\psi_0(x),
      \exp\!\left(\ii\left(\eta-\frac{\mu}{\omega t}\right)x\right)
      \psi_0\!\left(x+\frac{\xi}{\omega}
        +\frac{\nu}{\omega^2t}\right)\right\rangle.
  \end{eqnarray*}
  A straightforward computation then leads to the value
  (\ref{eq:phi}).
\end{proof}

Conversely, one can express $\mu$ and $\nu$ in terms of $\xi$, $\eta$
and $t$ provided $t$ is sufficiently small. In other words, one can
read equality (\ref{eq:expH_eq_expp_expx_expH}) from the right to the
left. The restriction on smallness of $t$ should not be considered as
surprising since the reversed equality is in fact an application of
the Baker-Campbell-Haussdorf formula which is known to be guaranteed
only locally.

\begin{alem}
  \label{thm:expp_expx_expH_eq_exp}
  For $\xi,\eta,t\in\R$, $|t|<2\pi/\omega$, it holds
  \begin{displaymath}
    \exp\!\left(\ii\,\frac{\xi}{\omega}\,p\right)
    \exp(\ii\eta x)\exp\!\left(-\ii tH_\omega\right)
    = e^{\ii\phi}\exp\!
    \left(-\ii tH_\omega+\ii\,\frac{\mu}{\omega}\,p+\ii\nu x\right)
  \end{displaymath}
  where
  \begin{displaymath}
    \mu = \frac{\omega t}{2}\left(\cot\!\left(
        \frac{\omega t}{2}\right)\xi+\eta\right),\textrm{~}
    \nu = \frac{\omega t}{2}\left(-\xi+\cot\!\left(
        \frac{\omega t}{2}\right)\eta\right),
  \end{displaymath}
  and
  \begin{equation}
    \label{eq:phi_eq_xieta}
    \phi = \frac{1}{2\omega}\,\xi\eta
    -\frac{\omega t-\sin(\omega t)}{8\omega
      \sin\!\left(\frac{\omega t}{2}\right)^2}
    \left(\xi^2+\eta^2\right).
  \end{equation}
\end{alem}

Set
\begin{displaymath}
  H(t) = H_\omega+f(t)x
\end{displaymath}
where the function $f(t)$ is supposed to be continuous and $T$
periodic. In \cite{enss_veselic83:_bound_propag_states_time_dep_hamil}
a formula has been derived for the propagator corresponding to the
Hamiltonian $H(t)$:
\begin{equation}
  \label{eq:propag_EnssVeselic}
  U(t,s) = \exp\!\big(-\ii\,\varphi_1(t,s)x\big)
  \exp\!\left(\ii\,\frac{\varphi_2(t,s)}{\omega}\,p\right)
  \exp\!\big(-\ii(t-s)H_{\omega}-\ii\psi(t,s)\big)
\end{equation}
where
\begin{eqnarray}
  \label{eq:varphi12}
  \varphi_1(t,s) &=& \int_s^t \cos(\omega(t-u))f(u)\,\dd u,
  \nonumber\\
  \varphi_2(t,s) &=& \int_s^t \sin(\omega(t-u)) f(u)\,\dd u,
\end{eqnarray}
and
\begin{equation}
  \label{eq:psi_def}
  \psi(t,s) = \frac{1}{2} \int_s^t
  \left(\varphi_1(v,s)^2-\varphi_2(v,s)^2\right)\,\dd v
\end{equation}
(be aware of a sign error in the definition of $\psi(t,s)$ in
\cite{enss_veselic83:_bound_propag_states_time_dep_hamil}). Our goal
is to transform formula (\ref{eq:propag_EnssVeselic}) due to Enss and
Veseli\`c into another one expressing the propagator as a single
exponential function of an operator.
\begin{aprop}
  \label{thm:propag_ultimite}
  For $t,s\in\R$, $(t-s)\notin(2\pi/\omega)\Z$, set
  \begin{displaymath}
    N = \left[\frac{\omega(t-s)}{2\pi}\right]\in\Z,\textrm{~}
    \Delta = \frac{2\pi}{\omega}
    \left\{\frac{\omega(t-s)}{2\pi}\right\}\in   
    \Big]0,\frac{2\pi}{\omega}\Big[              
  \end{displaymath}
  (where $[x]$ and $\{x\}$ are respectively the integer part and the
  fractional part of $x$). Then it holds
  \begin{equation}
    \label{eq:U_eq_exp}
    U(t,s) = (-1)^N\exp\!\left(-\ii\,\Delta H_{\omega}
      +\ii\,\frac{\mu(t,s)}{\omega}\,p
      +\ii\,\nu(t,s)x+\ii\sigma(t,s)\right)
  \end{equation}
  where
  \begin{eqnarray}
    \label{eq:mu_nu}
    \mu(t,s) &=& \frac{\omega\Delta}{2\sin\!
      \left(\frac{\omega(t-s)}{2}\right)}
    \int_s^t \sin\!\left(\omega\!
      \left(\frac{t+s}{2}-u\right)\right)f(u)\,\dd u,
    \nonumber\\
    \nu(t,s) &=& -\frac{\omega\Delta}{2\sin\!
      \left(\frac{\omega(t-s)}{2}\right)}
    \int_s^t \cos\!\left(\omega\!
      \left(\frac{t+s}{2}-u\right)\right)f(u)\,\dd u,
  \end{eqnarray}
  and
  \begin{eqnarray}
    \label{eq:sigma}
    \sigma(t,s) &=& -\frac{1}{2}
    \int_s^t \left(\varphi_1(v,s)^2-\varphi_2(v,s)^2\right)\,\dd v
    +\frac{1}{2\omega}\,\varphi_1(t,s)\varphi_2(t,s)
    \nonumber\\
    && -\,\frac{\omega\Delta-\sin(\omega(t-s))}{8\omega
      \sin\!\left(\frac{\omega(t-s)}{2}\right)^{\!2}}
    \left(\varphi_1(t,s)^2+\varphi_2(t,s)^2\right).
  \end{eqnarray}
\end{aprop}

\begin{proof}
  We have $t-s=(2\pi/\omega)N+\Delta$. Since the spectrum of
  $H_\omega$ equals $(\omega/2)+\omega\Z_+$ it holds
  \begin{displaymath}
    \exp\!\left(-\ii\,\frac{2\pi}{\omega}\,NH_\omega\right)
    = (-1)^N.
  \end{displaymath}
  Combining~(\ref{eq:propag_EnssVeselic}) with
  Lemma~\ref{thm:expp_expx_expH_eq_exp} we get
  \begin{eqnarray*}
    U(t,s) &=& (-1)^N\exp\!\big(-\ii\,\varphi_1(t,s)x\big)
    \exp\!\left(\ii\,\frac{\varphi_2(t,s)}{\omega}\,p\right)
    \exp\!\big(-\ii\Delta H_{\omega}-\ii\psi(t,s)\big) \\
    &=& (-1)^Ne^{\ii\phi(t,s)}
    \exp\!\left(-\ii\Delta H_\omega+\ii\,\frac{\mu(t,s)}{\omega}\,p
      +\ii\nu(t,s)x\right)
  \end{eqnarray*}
  where
  \begin{eqnarray*}
    \mu(t,s) &=& \frac{\omega\Delta}{2}\left(
      \cot\!\left(\frac{\omega\Delta}{2}\right)\varphi_2(t,s)
      -\varphi_1(t,s)\right), \\
    \nu(t,s) &=& \frac{\omega\Delta}{2}\left(\varphi_2(t,s)
      +\cot\!\left(\frac{\omega\Delta}{2}\right)\varphi_1(t,s)\right),
  \end{eqnarray*}
  and
  \begin{displaymath}
    \phi(t,s) = \frac{1}{2\omega}\,\varphi_1(t,s)\varphi_2(t,s)
    -\frac{\omega\Delta-\sin(\omega\Delta)}{8\omega
      \sin\!\left(\frac{\omega\Delta}{2}\right)^2}
    \left(\varphi_1(t,s)^2+\varphi_2(t,s)^2\right)-\psi(t,s).
  \end{displaymath}
  After some simple manipulations one arrives at the desired formula
  (\ref{eq:U_eq_exp}).
\end{proof}

\end{document}